\documentclass[runningheads]{llncs}

\usepackage[T1]{fontenc}
\usepackage{graphicx}

\usepackage{hyperref}
\usepackage{amsmath}
\usepackage{amssymb}
\usepackage{enumerate}
\usepackage{xspace}
\usepackage{amscd}
\usepackage{tikz-cd}
\usepackage{algorithm}
\usepackage{algpseudocode}
\spnewtheorem{heuristic}{Heuristic}{\bfseries}{\itshape}
\spnewtheorem{expfact}{Experimental fact}{\bfseries}{\itshape}
\spnewtheorem{unproven}{Unproven assertion}{\itshape}{\upshape}

\let\phi\varphi
\let\epsilon\varepsilon
\let\subset\subseteq
\let\supset\supseteq
\let\subsetneq\varsubsetneq

\DeclareMathOperator{\im}{im}
\DeclareMathOperator{\rk}{rk}
\DeclareMathOperator{\coker}{coker}

\DeclareMathOperator{\Sym}{Sym}

\newcommand{\PP}{\mathbf{P}}
\newcommand{\corps}{\mathbb}

\newcommand{\F}{\corps{F}}
\newcommand{\K}{\corps{K}}

\newcommand{\longto}{\longrightarrow}

\newcommand{\surj}{\twoheadrightarrow}
\newcommand{\tens}{\otimes}

\newcommand{\card}{\#}

\newcommand{\linspan}[2][]{\langle #2\rangle_{#1}}
\newcommand{\deux}[1][2]{^{\langle #1\rangle}}
\newcommand{\scal}[2]{\langle #1,#2\rangle}
\newcommand{\id}{\operatorname{id}}
\newcommand{\ev}{\operatorname{ev}}
\newcommand{\ind}{\operatorname{ind}}
\newcommand{\defect}{\operatorname{def}}
\newcommand{\w}{\operatorname{w}}
\newcommand{\dist}{\operatorname{d}}
\newcommand{\dmin}{\dist_{\min}}
\newcommand{\rmax}{r_{\max}}
\newcommand{\rmin}{r_{\min}}
\newcommand{\rcrit}{r_{\operatorname{crit}}}
\newcommand{\code}[1][C]{{\fontfamily{lmss}\selectfont\mathsf{#1}}}
\newcommand{\Alt}{\operatorname{\code[Alt]}}
\newcommand{\Gop}{\operatorname{\code[Gop]}}
\newcommand{\GRS}{\operatorname{\code[GRS]}}
\newcommand{\Classic}{\texttt{Classic McEliece}\xspace}

\newcommand{\mot}{\mathbf}
\newcommand{\mat}{\mathbf}

\newcommand{\beq}{\begin{equation}}
\newcommand{\eeq}{\end{equation}}

\begin{document}

\title{The syzygy distinguisher\thanks{This is an updated and expanded version of the work with the same title presented at Eurocrypt 2025 \cite{EC:Randriambololona25}. Details about the new material can be found in the Introduction.}}

\subtitle{\scriptsize{\normalfont(version to be submitted to Journal of Cryptology)\\ \vspace{-.9\baselineskip}}}

\author{Hugues Randriambololona}
\institute{ANSSI, Laboratoire de cryptographie\\ \& Télécom Paris, C\textsuperscript{2}\\ \email{hugues.randriam@ssi.gouv.fr}}

\maketitle

\begin{abstract}
We present a new distinguisher for alternant and Goppa codes,
whose complexity is subexponential in the error-correcting capability,
hence better than that of generic decoding algorithms. Moreover it does not suffer
from the strong regime limitations of the previous distinguishers or structure recovery
algorithms: in particular, it applies to the codes used in the \Classic candidate
for postquantum cryptography standardization.
The invariants that allow us to distinguish are graded Betti numbers of the homogeneous
coordinate ring of a shortening of the dual code.

Since its introduction in 1978, this is the first time an analysis (in the CPA model)
of the McEliece cryptosystem breaks the exponential barrier.
\end{abstract}

\section{Introduction}

In the McEliece cryptosystem \cite{perso:McEliece78}, a private message is encoded
as a codeword in a public binary Goppa code \cite{Goppa70}, with some noise added. Knowing the secret algebraic
data that served to construct the public code, the legitimate recipient has an efficient
decoding algorithm and can recover the message. However, to an attacker, the public code
looks like a random code, and removing the noise is untractable.

At the time of publication of \cite{perso:McEliece78} it was not customary to fully formalize the
expected security properties of the scheme. In modern language we would say that it aims
at one-wayness against chosen plaintext attacks (OW-CPA). There are now variants of the McEliece scheme
that claim stronger security properties (e.g. IND-CCA2), but for this they rely on the hypothesis
that the original McEliece scheme is OW-CPA. This makes it even more relevant today to scrutinize whether
this hypothesis can be reduced to more fundamental theoretical assumptions.

And indeed, as the informal description of the system naturally suggests, a traditional security proof
reduces this hypothesis to the following two assumptions:
\begin{enumerate}
\item\label{indistinguishability} Goppa codes are computationally indistinguishable from generic linear codes
(say, when described by generator matrices in reduced row echelon form).
\item\label{decoding} Decoding a generic linear code is difficult.
\end{enumerate}

Cryptanalytic attempts can be classified depending on whether they target assumption \ref{indistinguishability}
or \ref{decoding}.

First, those aiming at assumption~\ref{indistinguishability} themselves come in two flavours:
\begin{itemize}
\item \emph{Distinguishers} address the decisional version of the problem: given a generator matrix,
decide whether it is that of a Goppa code or a generic code.
\item \emph{Key recovery attacks} address the computational version:
recover the Goppa structure of the code, or at least an equivalent one \cite{EC:Gibson91}, if it exists.
\end{itemize}
Assumption~\ref{indistinguishability} was first introduced formally in \cite{AC:CouFinSen01}.
Although certain arguments, such as the fact that the class of Goppa codes is very large,
make this assumption plausible, it remains quite ad hoc from a theoretical
point of view. Its sole virtue is that it passed the test of time.
There seems to be something special with Goppa codes happening there: indeed, variants of the
McEliece system were proposed, with Goppa codes replaced with other types of codes allowing
more manageable parameters; however, most of these propositions were eventually broken, as the
hidden structure of the codes could be recovered.
Also, although the McEliece cryptosystem marked the birth of code-based cryptography,
the idea of having an object (such as a Goppa code) constructed from data defined over an extension
field and then masked by considering it over a small subfield, was then found in other branches
of cryptography: for instance it is at the basis of the HFE cryptosystem \cite{EC:Patarin96}
in multivariate cryptography.
Such systems can often be attacked by algebraic methods
(including, but not limited to, the use of Gröbner basis algorithms).
This suggests that if a weakness in assumption~\ref{indistinguishability} were to be uncovered one day,
then algebraic methods should be a tool of choice.
However, up to now, the best distinguishers and key recovery algorithms such as those
in \cite{FGOPT13}\cite{EC:CouOtmTil14}\cite{AC:CouMorTil23}\cite{BarMorTil24}\cite{EPRINT:LemMorTil25}\cite{tangent}
only apply to alternant or Goppa codes with very degraded parameters.
Against McEliece with cryptographically relevant parameters, either they have exponential complexity with large
constants, which makes them useless, or worse: they simply cease to work.

Assumption~\ref{decoding}, on the other hand, stands on a firm theoretical ground:
the decoding problem for generic linear codes is known to be NP-hard \cite{BMT78}.
As such, it is believed to resist the advent of quantum computers, which made
the McEliece system a good candidate for postquantum cryptography.
However, starting with Prange's information set decoding algorithm \cite{Prange62},
generic decoding methods saw continuous incremental improvements over time.
Joint with technological progress in computational power, this eventually led
to practical \emph{message recovery attacks} against the McEliece system with its
initially proposed parameter set.
However, it appears that this weakness was only the result of a too optimistic
choice of parameters.
With the need for new standards for postquantum cryptography,
an updated version named \Classic was proposed, still relying on binary Goppa codes,
but with more conservative parameter sets
adapted to resist the best generic decoding attacks with some margin of safety.
In the design rationale of this new system \cite{rationale22} one can find
an impressive list of several dozens of papers on generic decoding algorithms,
ranging over the last five decades. It is then observed that all these algorithms
have complexity exponential in the error-correcting capability of the code. Better, the constant in this
exponential is still the same as in Prange's original result: all improvements
remain confined in terms of lesser order!
This might give a feeling that we could possibly have reached the intrinsic
complexity of cryptanalysis of this system.

Preventing the spread of this belief, we present a new distinguisher for alternant and Goppa codes,
i.e. a basic structural analysis of the McEliece cryptosystem,
whose asymptotic complexity is subexponential in the error-correcting capability,
hence better than that of generic decoding algorithms.
Moreover, for given finite parameters, it is more efficient than state-of-the-art distinguishers or key recovery
methods, while not suffering from their strong regime limitations:
in particular, it applies to the codes used in \Classic.

\subsection*{Principles and organization}

A natural strategy to build a distinguisher is to design code invariants --- quantities
that intrinsically depend only on the code, not on the choice of a generator matrix --- that
behave differently for the classes of codes we want to distinguish.
The invariants we use here are graded Betti numbers of the homogeneous
coordinate ring of a shortening of the dual code.
Generators of the dual of an alternant or Goppa code, after extension of scalars,
satisfy quadratic relations of a special form: they can be expressed as $2\times 2$ minors
of a matrix. As such, we can find relations between these quadratic relations, called \emph{syzygies}.
Then these syzygies also satisfy relations, and iterating this process we get higher syzygies up to
some order that we can estimate. However, for generic codes, we do not expect this to
happen in the same magnitude. This directly gives a distinguisher, at least in theory.
In practice, computing these spaces of syzygies only involves basic linear algebra:
they can be constructed iteratively, as the kernel of some generalized Macaulay matrices.
This can be done efficiently, except for the fact that the dimension of the spaces involved grows
exponentially.
Our last ingredient is then \emph{shortening}, which allows us to work with syzygies of a lesser order,
and keep these dimensions more under control.

Let us quickly illustrate our result with two basic examples.

First, our distinguisher can be seen as a generalization of the so-called \emph{square code} distinguisher,
first presented in \cite{FGOPT13}, reinterpreted in \cite{MarPel12}, and fully analyzed in \cite{MorTil23}.
Let $\code$ be a $[n,k]$-code, and $S_2$ the space of quadratic forms in $k$ indeterminates.
Let
\beq
\ev_2: S_2\longto\F^n
\eeq
be the evaluation map at the columns of a given generator matrix of $\code$.
Then the image of $\ev_2$ is the square code $\code\deux$, and its kernel is the space of quadratic relations $I_2(\code)$.
The dimensions of these spaces are related, and can be expressed as a Betti number:
\beq
\beta_{1,2}(\code)=\dim(I_2(\code))=\binom{k+1}{2}-\dim(\code\deux).
\eeq
Now \cite{FGOPT13} gives a lower bound on this $\beta_{1,2}(\code)$ when $\code$ is the dual of an alternant or (binary) Goppa code.
On the other hand, \cite{CCMZ15} shows $\beta_{1,2}(\code)=\left(\binom{k+1}{2}-n\right)^+$ with high probability when $\code$ is random.
If this quantity is smaller than the said lower bound, we can distinguish.

Likewise we claim that Theorem~2.8 of \cite{Eisenbud05} provides a $\beta_{2,3}$-based distinguisher for $\GRS$ codes
among MDS $[7,4]$-codes.
Indeed, for such codes, the square always fills the whole space --- so the square distinguisher does not apply ---
and we have $\dim(I_2(\code))=\binom{4+1}{2}-7=3$.
Let $Q_1,Q_2,Q_3$ be a basis of $I_2(\code)$.
By definition, $Q_1,Q_2,Q_3$ do not satisfy linear relations with coefficients in $\F$, however they can satisfy
relations whose coefficients are forms of degree $1$.
Such relations are also called degree~$3$ syzygies.
They live in the kernel of the degree~$3$ Macaulay matrix
\beq
\mat{M}_3: I_2(\code)\tens S_1\longto S_3
\eeq
where $S_1$ (resp. $S_3$) is the space of homogeneous linear (resp. cubic) forms in $k=4$ variables.
Now Theorem~2.8 of \cite{Eisenbud05} shows that for a non-$\GRS$ code the map $\mat{M}_3$ is injective.
On the other hand, if $\code$ is $\GRS$ with standard basis $\mot{y},\mot{y}\mot{x},\mot{y}\mot{x}^2,\mot{y}\mot{x}^3$
then we can take
\beq
Q_1=X_1X_3-X_2^2,\quad Q_2=X_1X_4-X_2X_3,\quad Q_3=X_2X_4-X_3^2
\eeq
and these satisfy the syzygies
\beq
X_1Q_3-X_2Q_2+X_3Q_1 \;=\; X_2Q_3-X_3Q_2+X_4Q_1 \;=\; 0.
\eeq
Thus we can distinguish by computing $\beta_{2,3}(\code)=\dim\ker(\mat{M}_3)$,
which will yield~$2$ for $\GRS$ and~$0$ for non-$\GRS$ MDS $[7,4]$-codes.

We generalize these examples to higher Betti numbers, following the exact same pattern:
\begin{itemize}
\item On one hand, we give lower bounds on the Betti numbers of algebraic codes (dual alternant, dual Goppa, and their shortened subcodes).
This is done in section~\ref{secEN}, using the Eagon-Northcott complex, a tool precisely crafted to detect long conjugate $\GRS$ subcodes.
\item On the other hand, we estimate the Betti numbers of random codes in terms of raw code parameters (length, dimension, distance).
This is done in section~\ref{secdefect}, partially relying on a natural heuristic: random codes are not expected to admit more syzygies than those forced by these parameters.
We do not have full proofs for this fact, but we provide experimental evidence and partial theoretical arguments that support it.
\end{itemize}
Prior to that, section~\ref{secBetti} explains how these invariants can be effectively computed.
Last, section~\ref{paramdist} combines everything and chooses parameters in order to optimize asymptotic complexity.
Of special importance in this regard is Proposition~\ref{ENrac} from section~\ref{secEN},
which shows that the property that the ideal of a code contains minors of a matrix of linear forms passes to its shortened subcodes.

\subsection*{Related (and unrelated) works}

As already noted, our distinguisher can be seen as a generalization of the square distinguisher of \cite{FGOPT13}.
Using an approach similar to ours, the work \cite{AC:CouMorTil23} also extends this square distinguisher
by exploiting special properties of the space of quadratic relations, but in a different direction:
the authors detect the existence of short relations by computing the Hilbert series of a convenient Pfaffian modeling
(and in some cases they even manage to recover the key).
Last, the key recovery attack in \cite{BarMorTil24} combines shortening of the dual code
with ideas from the square distinguisher to reduce parameters, and then performs a careful algebraic modeling exploiting
the symmetries of a system, that can finally be solved with Gröbner basis algorithms.
All these results have limited range of applicability, but they introduced numerous techniques
that influenced the present work.

After the first version of this text appeared, other works continued this series of results.
A refined analysis of the Pfaffian distinguisher can now be found in~\cite{EPRINT:LemMorTil25}.
Using purely geometric methods, \cite{tangent} presents a new key recovery attack; but quite
surprisingly, this attack works precisely in the regimes where those of \cite{AC:CouMorTil23} and \cite{BarMorTil24} apply!
Last, the recent work \cite{multiplicites} proposes a distinguisher based on higher degree relations vanishing with
higher multiplicities.

From a geometric point of view, the Betti numbers and the syzygies we consider are those
of a set of points in projective space (namely, defined by the columns of a parity check matrix of the code).
As such, they have been already extensively studied. Of notable importance to us are the works
\cite{GreLaz88}\cite{Lorenzini93}\cite{HirSim96}\cite{EisPop99},
in relation with the so-called \emph{minimal resolution conjecture} --- regardless of it being false in general:
we just request it being ``true enough''.
Initially, syzygies of sets of points were considered as a mere tool in the study of syzygies of curves.
They were then studied for themselves, but the focus was mostly on points in generic position,
over an algebraically closed field.
Keeping applications to coding theory and cryptography in mind, we will have more interest
in finite field effects.

Syzygies sometimes appear as a tool in cryptanalytic works, or in the study of Gröbner basis algorithms;
however in general only the first module of syzygies is considered, not those of higher order.
Likewise, a few works in coding theory (such as \cite{RenTap96} or \cite{Hansen03}) use homological
properties of finite sets of points; but the applications differ from ours.

Last, note that our approach is apparently unrelated to the series of works initiated with \cite{JohVer14}:
while these authors also define Betti numbers for codes, these are constructed
from the Stanley-Reisner ring of the code matroid, not the homogeneous coordinate ring.
This leads to different theories, although seeking links between the two could be an interesting project.

\subsection*{Additional material and new results since the Eurocypt 2025 version}

There are significative differences between this text and its prior version~\cite{EC:Randriambololona25}.
First, the supplementary material to~\cite{EC:Randriambololona25} is now included in the main text,
sometimes after some reworking.
Moreover, several new results were introduced.

\begin{enumerate}[(a)]
\item The discussion covering possible further improvements in the description of
the minimal resolution of algebraic codes has been incorporated in section~\ref{secEN}.
After some reworking, it now includes Corollary~\ref{corV2} that provides an explicit
family of linearly independent quadratic relations for dual alternant codes that matches
the lower bound of~\cite[Th.~20]{MorTil23} (thus giving an alternative
proof for this result while sticking to the approach of~\cite{FGOPT13}).
\item\label{newrmax} A tight lower bound on the number of columns of a matrix whose $2\times 2$ minors
give quadratic relations (and determine the $\rmax$) of Goppa codes is now stated
in Remark~\ref{improvedcorlowerboundGop} in section~\ref{secEN} (although its proof
is deferred to the forthcoming work~\cite{stability}).
\item The computation of the Betti diagram of $[k,k+1]$ MDS codes and of $[k,2k-1]$ $\GRS$ codes
has been incorporated in section~\ref{secdefect}.
\item The presentation of experimental data on defects of random codes also has been incorporated in section~\ref{secdefect}.
\item The new Proposition~\ref{rmax>k+1-d} in section~\ref{secdefect} proves a (very small) part of Experimental fact~\ref{Betti_dmin_dduale}, linking the $\rmax$ of a code to its minimum distance.
\item The two examples of application of the distinguisher to some Goppa codes with small parameters have been incorporated in section~\ref{paramdist}.
\item\label{newKW} A new variant of the distinguisher, based on Koszul cohomology and the (block) Wiedemann algorithm, has been introduced also in section~\ref{paramdist}.
\end{enumerate}
Some of these new results allow a noticeable improvement on the complexity of the distinguisher.
Point~\eqref{newrmax} allows to predict the optimal number $s$ of coordinates to shorten.
Point~\eqref{newKW} allows to use linear algebra with exponent $2$ instead of $\omega$
(and also to relax the dependency of the complexity estimate from some of the heuristics).
When applied e.g. to the \Classic 348864 parameters, each allows to gain roughly 60 bits of complexity,
and combined, they reduce the complexity estimate from $2^{528}$ to $2^{401}$.

\subsection*{Notation and conventions}

We use row vector convention.
We try to consistently use
lowercase bold font for codewords and vectors: $\mot{c}$, $\mot{x}$, $\mot{y}$;
uppercase bold for matrices: $\mat{G}$, $\mat{H}$, $\mat{M}$;
sans-serif for codes: $\code$, $\GRS$, $\Alt$, $\Gop$.

The book \cite{Eisenbud05} will be our main source on syzygies.
For codes, especially the link between powers of codes and the geometric view on coding theory, we will follow \cite{HR15}.

Given a field $\F$, we see $\F^n$ as the standard product algebra of dimension~$n$.
Thus $\F^n$ is not a mere vector space, it comes canonically equipped with componentwise
multiplication: for $\mot{x}=(x_1,\dots,x_n)$ and $\mot{y}=(y_1,\dots,y_n)$ in $\F^n$,
\beq\mot{x}\mot{y}=(x_1y_1,\dots,x_ny_n).\eeq
(Some authors call this the Schur product of $\mot{x}$ and $\mot{y}$;
how the name of this great mathematician ended associated with this trivial operation
is quite \emph{convoluted}.)

In any algebra, we can define a trace bilinear form. In the case of $\F^n$, this trace
bilinear form is the standard scalar product:
\beq\scal{\mot{x}}{\mot{y}}=x_1y_1+\cdots+x_ny_n.\eeq

A $k$-dimensional subspace $\code\subset\F^n$ is called a $[n,k]$-code (and a $[n,k]_q$-code in case $\F=\F_q$).
The orthogonal space $\code^\perp$ is called the \emph{dual code} of $\code$. 

Componentwise multiplication extends to codes, taking the linear span:
for $\code,\code'\subset\F^n$,
\beq\code\code'=\linspan[\F]{\mot{c}\mot{c'}:\:\mot{c}\in\code,\,\mot{c'}\in\code'}.\eeq
Powers $\code\deux[r]$ of a code are defined inductively: $\code\deux[0]=\F\cdot\mot{1}$
is the $1$-dimensional repetition code, and $\code\deux[r+1]=\code\deux[r]\code$.


If $\code$ is a $[n,k]$-code, a generator matrix for $\code$ is a $k\times n$ matrix $\mat{G}$
whose rows $\mot{c}_1,\dots,\mot{c}_k$ form a basis of $\code$.
A parity check matrix $\mat{H}$ for $\code$ is a generator matrix for $\code^\perp$.

Thanks to the algebra structure, polynomials in one or several variables can be
evaluated in $\F^n$. In particular, let
\beq S=\F[X_1,\dots,X_k] \eeq
be the algebra of polynomials in $k$ variables over $\F$, graded by total degree.
Evaluation at the rows $\mot{c}_1,\dots,\mot{c}_k$ of $\mat{G}$ then gives linear map
\beq \ev_{\mat{G}}: S \longto \F^n. \eeq
Observe that if $\mot{p}_1,\dots,\mot{p}_n$ are the columns of $\mat{G}$,
then for $f(X_1,\dots,X_k)\in S$ we have
\beq
\begin{split}
\ev_{\mat{G}}(f)&=f(\mot{c}_1,\dots,\mot{c}_k)\\
&=(f(\mot{p}_1),\dots,f(\mot{p}_n))
\end{split}
\eeq
where in the first line we have one evaluation of $f$ at a $k$-tuple of vectors,
while in the second line we have a vector of evaluations of $f$ at $k$-tuples of scalars.

A code is projective if it has dual minimum distance $\dmin(\code^\perp)\geq3$,
or equivalently if no two of the $\mot{p}_i$ are proportional. Any code can be ``projectivized''
by discarding (puncturing) coordinates, keeping only one $\mot{p}_i$ in each nonzero proportionality class.

Restricting to homogeneous polynomials of degree $r$, we have a \emph{surjective} map $S_r\longto\code\deux[r]$,
whose kernel we denote $I_r(\code)$.
We then define the homogeneous coordinate ring of $\code$ as the formal direct sum
\beq
\code\deux[\cdot]=\bigoplus_{r\geq0}\code\deux[r]
\eeq
and likewise its homogeneous ideal
\beq
I(\code)=\bigoplus_{r\geq0}I_r(\code).
\eeq
It turns out these are also the homogeneous coordinate ring and the homogeneous ideal of the finite set of points
\beq
\{\mot{p}_1,\dots,\mot{p}_n\}\subset\PP^{k-1}.
\eeq
The short exact sequence
\beq
0\longto I(\code)\longto S\longto \code\deux[\cdot]\longto 0
\eeq
makes $\code\deux[\cdot]$ a homogeneous quotient ring of $S$.
In this work we will use coordinates, but identifying $S$ with the symmetric algebra of $\code$
would allow to make all this coordinatefree.

Given $\mot{x},\mot{y}\in\F^n$,
all entries of $\mot{x}$ distinct, all entries of $\mot{y}$ nonzero,
the \emph{generalized Reed-Solomon code} of order $k$ with support vector $\mot{x}$
and multiplier $\mot{y}$ is
\beq\GRS_k(\mot{x},\mot{y})=\linspan[\F]{\mot{y},\mot{y}\mot{x},\dots,\mot{y}\mot{x}^{k-1}}=\{\mot{y}f(\mot{x}):\:f(X)\in\F[X]_{<k}\}\subset\F^n.\eeq
It is a $[n,k]$-code if $k\leq n$.

Now let $\F_q\subset\F_{q^m}$ be an extension of finite fields.
Given $\mot{x},\mot{y}\in(\F_{q^m})^n$ satisfying the same conditions as above,
the \emph{alternant code} of order $t$ and extension degree $m$ over $\F_q$,
with support $\mot{x}$ and multiplier $\mot{y}$, is
\beq\begin{split}
\Alt_t(\mot{x},\mot{y})&=\GRS_t(\mot{x},\mot{y})^\perp\cap(\F_q)^n\\
&=\{\mot{c}\in(\F_q)^n:\:c_1y_1x_1^j+\cdots+c_ny_nx_n^j=0\quad(0\leq j<t)\},
\end{split}\eeq
with parameters $[n,(\geq)n-mt]_q$.

Last, given a polynomial $g(X)\in\F_{q^m}[X]$ that does not vanish on any entry of $\mot{x}$,
the $q$-ary \emph{Goppa code} with support $\mot{x}$ and Goppa polynomial $g$ is
\beq\Gop(\mot{x},g)=\Alt_{\deg(g)}(\mot{x},g(\mot{x})^{-1}).\eeq

We will work mostly in the class
\beq\Alt^\perp_{q,m,n,t}\eeq
of \emph{dual} $q$-ary alternant codes of extension degree $m$, length $n$, and order $t$.
If $q$ is unspecified we take $q=2$. If $n$ is unspecified we take $n=q^m$.
We say a code $\code\in\Alt^\perp_{q,m,n,t}$ is \emph{proper} if it has dimension
\beq k=mt.\eeq
In this case, after extension of scalars, we have
\beq\label{decompose_alternant}
\code_{\F_{q^m}}=\GRS_t(\mot{x},\mot{y})\oplus\GRS_t(\mot{x}^q,\mot{y}^q)\oplus\cdots\oplus\GRS_t(\mot{x}^{q^{m-1}},\mot{y}^{q^{m-1}}).
\eeq
Likewise we define the corresponding class
\beq\Gop^\perp_{q,m,n,t}\eeq
of dual Goppa codes, and $\code\in\Gop^\perp_{q,m,n,t}$ is said proper if it is when seen in $\Alt^\perp_{q,m,n,t}$.
Also we define subclasses $\Gop^{\operatorname{irr},\perp}_{q,m,n,t}\subset\Gop^{\operatorname{sqfr},\perp}_{q,m,n,t}\subset\Gop^\perp_{q,m,n,t}$,
in which the Goppa polynomial is irreducible or squarefree, respectively.

\section{Minimal resolutions and graded Betti numbers}\label{secBetti}

\subsection*{Generalities}

We freely borrow results and terminology from \cite{Eisenbud05}, and then elaborate on the parts of the theory that we will need.

Let $S=\F[X_1,\dots,X_k]$ be the $k$-dimensional polynomial ring over $\F$, graded by total degree.
If $M_0$ is a finitely generated graded $S$-module, and $F_0$ is the free module on a minimal system
of homogeneous generators of $M_0$, then the (first) syzygy module of $M_0$ is $M_1=\ker(\;F_0\longto M_0\;)$.
Concretely, if $g_1,\dots,g_N$ are minimal generators of $M_0$, elements of $F_0$ can be seen as formal
sums $\sum_uf_u[g_u]$ with $f_u\in S$, where the $[g_u]$ are just formal symbols;
such a formal sum then lies in $M_1$ when the actual sum evaluated in $M_0$ satisfies $\sum_uf_ug_u=0$.

Iterating this construction, we obtain a minimal resolution
\beq
\cdots \longto F_2 \longto F_1 \longto F_0
\eeq
of $M_0$,
where the graded free modules $F_i$, together with the iterated syzygy modules $M_i$, are constructed inductively:
\begin{itemize}
\item $F_i$ is the free module on a minimal system of homogeneous generators of $M_i$
\item $M_{i+1}=\ker(\;F_i\longto M_i\;)$.
\end{itemize}
The number $\beta_{i,j}$ of degree $j$ elements in a minimal system of generators of $M_i$ does not depend on
the choices made, and is called the $(i,j)$-th graded Betti number of $M_0$.
Keeping track of the grading, we have
\beq
F_i=\bigoplus_{j\geq0}S(-j)^{\beta_{i,j}}
\eeq
where $S(-j)$ is the free rank $1$ module generated in degree $j$, so $S(-j)_d=S_{d-j}$.

It is customary to display the graded Betti numbers in the form of a Betti diagram, as follows:
\begin{equation*}
\arraycolsep=8pt
\begin{array}{c|ccccc}
& 0 & 1 & \dots & i & \dots\\
\hline
\vdots & \vdots & \vdots & & \vdots & \\
r & \beta_{0,r} & \beta_{1,r+1} & \dots & \beta_{i,r+i} & \dots \\ 
r+1 & \beta_{0,r+1} & \beta_{1,r+2} & \dots & \beta_{i,r+i+1} & \dots \\
\vdots & \vdots & \vdots & & \vdots & \\
\end{array}
\end{equation*}
with null entries marked as ``$-$'' for readability.

\begin{lemma}\label{quadrant}
Assume $M_i$ is generated in degrees $\geq D$, i.e. $\beta_{i,j}=0$ for all $j\leq D-1$.
Then $M_{i+1}$ is generated in degrees $\geq D+1$,
and by induction all the upper-right quadrant of the Betti diagram defined by $\beta_{i,D-1}$ vanishes.
\end{lemma}
\begin{proof}
Let $g_1,\dots,g_N$ form a minimal system of generators of $M_i$.
By hypothesis all $g_u$ have degree $\geq D$.
Then, in a homogeneous relation $\sum_uf_ug_u=0$, no $f_u$ can be a nonzero constant:
otherwise the corresponding $g_u$ could be expressed as a linear combination of the others,
hence could be removed from the system of generators, contradicting minimality.
Thus all nonzero $f_u$ have degree $\geq1$,
and the element of $M_{i+1}$ defined by this homogeneous relation has degree $\geq D+1$.
\end{proof}

As $M_i$ and $F_i$ are graded $S$-modules, we will also write $M_{i,j}$ and $F_{i,j}$
for their degree $j$ homogeneous components.

From now on let $\code$ be a $[n,k]$-code, with generator matrix $\mat{G}$,
and $M_0=\code\deux[\cdot]$ its homogeneous coordinate ring.
Then $M_0$ is a dimension~$1$ Cohen-Macaulay quotient of $S$,
and by the Auslander-Buchsbaum formula its minimal resolution has length $k-1$.
So we have an exact sequence
\beq
0\longto F_{k-1}\longto \cdots \longto F_2 \longto F_1 \longto F_0=S \longto \code\deux[\cdot] \longto 0
\eeq
and the Betti diagram has $k$ columns, indexed from $0$ to $k-1$.

As $M_0=\code\deux[\cdot]$ and $F_0=S$, we have $\beta_{0,0}=1$ and $\beta_{0,j}=0$ for $j\neq0$,
i.e. the $0$-th column of the Betti diagram is always $(1,-,-,\dots)^\top$.

Likewise $M_1=I(\code)$, 
and $\beta_{1,j}$ is the number of homogeneous polynomials of degree $j$ in a minimal system of generators of $I(\code)$.
As the evaluation map $S_1\longto C$ is an isomorphism, we see that $I(\code)$ is generated
in degrees $\geq2$. Lemma~\ref{quadrant} then implies:
\begin{lemma}\label{geni+1}
For any $i\geq1$, the $i$-th syzygy module $M_i$ of $\code$ is generated in degrees $\geq i+1$.
Thus we have $F_{i,j}=M_{i,j}=0$ for $j\leq i$, and
\beq\label{F=M}
F_{i,i+1}=M_{i,i+1}=\F^{\beta_{i,i+1}}.
\eeq
\end{lemma}
Hence the $0$-th row of the Betti diagram is always $(1,-,-,\dots)$.

Figures \ref{figureHamming}-\ref{figurebinaryGolay} offer a few such diagrams for contemplation.

\begin{figure}[!h]
\vspace{-.5\baselineskip}
\begin{minipage}[b]{0.49\textwidth}
\begin{equation*}
\arraycolsep=5pt
\begin{array}{c|cccc}
& 0 & 1 & 2 & 3\\
\hline
0 & 1 & - & - & -\\
1 & - & 3 & - & -\\
2 & - & 1 & 6 & 3\\
\end{array}
\end{equation*}
\vspace{-\baselineskip}\caption{the $[7,4]_2$ Hamming code}\label{figureHamming}
\end{minipage}
\begin{minipage}[b]{0.5\textwidth}
\begin{equation*}
\arraycolsep=5pt
\begin{array}{c|cccccc}
& 0 & 1 & 2 & 3 & 4 & 5\\
\hline
0 & 1 & - & - & - & - & -\\
1 & - & 10 & 16 & - & - & -\\
2 & - & 1 & 5 & 26 & 20 & 5\\
\end{array}
\end{equation*}
\vspace{-\baselineskip}\caption{the $[11,6]_3$ Golay code}\label{figureternaryGolay}
\end{minipage}
\vspace{.5\baselineskip}
\begin{equation*}
\arraycolsep=5pt
\begin{array}{c|cccccccccccc}
& 0 & 1 & 2 & 3 & 4 & 5 & 6 & 7 & 8 & 9 & 10 & 11\\
\hline
0 & 1 & - & - & - & - & - & - & - & - & - & - & -\\
1 & - & 55 & 320 & 891 & 1408 & 1210 & 320 & 55 & - & - & - & -\\
2 & - & 1 & 11 & 55 & 220 & 650 & 1672 & 1870 & 1221 & 485 & 110 & 11\\
\end{array}
\end{equation*}
\vspace{-\baselineskip}\caption{the $[23,12]_2$ Golay code}\label{figurebinaryGolay}
\vspace{-1.5\baselineskip}
\end{figure}

Observe that replacing $\mat{G}$ with $\mat{S}\mat{G}$, for $\mat{S}\in\F^{k\times k}$ an invertible matrix,
corresponds to a linear change of coordinates for the variables $X_1,\dots,X_k$.
This is an automorphism of $S$, hence will not affect the Betti numbers.
Thus these Betti numbers really are invariants of $\code$, independently of the choice of $\mat{G}$.

More generally recall that two $[n,k]$ codes $\code_1,\code_2$, defined by generator matrices $\mat{G}_1,\mat{G}_2$,
are linearly isometric, or monomially equivalent, if $\mat{G}_1=\mat{S}\mat{G}_2\mat{P}\mat{D}$
for $\mat{S}\in\F^{k\times k}$ an invertible matrix, and $\mat{P},\mat{D}\in\F^{n\times n}$ a permutation matrix
and an invertible diagonal matrix.
\begin{lemma}\label{invBetti}
Monomially equivalent codes have the same Betti numbers.
\end{lemma}
\begin{proof}
We already discussed the action of $\mat{S}$. As for the action of $\mat{P}\mat{D}$,
it preserves the kernel of the graded evaluation map $S\longto\code\deux[\cdot]$.
\end{proof}

\subsection*{Effective computation of linear syzygies}

There are several algorithms to compute minimal resolutions and graded Betti numbers in general.
Many of them rely first on a Gröbner basis computation.
However in this work we will only be interested in computing the first (nontrivial) row of the Betti diagram,
or equivalently, the so-called \emph{linear strand} of the resolution.
This easier computation can be described in elementary terms.
First, for any $r\geq3$, consider the natural multiplication map
\beq\label{defphir}
\phi_r: M_{r-2,r-1}\tens S_1 \longto M_{r-2,r}.
\eeq
\begin{lemma}\label{kercokerphir}
We have:
\begin{align}
\label{kerphir} \ker(\phi_r)=M_{r-1,r}&\simeq\F^{\beta_{r-1,r}}\\
\label{cokerphir} \coker(\phi_r)&\simeq\F^{\beta_{r-2,r}}.
\end{align}
\end{lemma}
\begin{proof}
We prove \eqref{cokerphir} first.
Let $\mathcal{G}$ be a minimal system of homogeneous generators of $M_{r-2}$, and
for each $j$ let $B_{r-2,j}\subset M_{r-2,j}$ be the linear subspace generated by the degree $j$ elements of $\mathcal{G}$.
Then, as $M_{r-2}$ is generated in degrees $\geq r-1$, we have $B_{r-2,r-1}=M_{r-2,r-1}$,
and $B_{r-2,r}$ is a complementary subspace to $S_1\cdot B_{r-2,r-1}=\im(\phi_r)$ in $M_{r-2,r}$. This proves \eqref{cokerphir}.

Now let $F_{r-2}$ be the free graded module on $\mathcal{G}$.
Then $M_{r-1,r}$ is the kernel of the natural map $F_{r-2,r}\longto M_{r-2,r}$.
However, under the decompositions $F_{r-2,r}=(B_{r-2,r-1}\tens S_1)\oplus B_{r-2,r}$ and $M_{r-2,r}=\im(\phi_r)\oplus B_{r-2,r}$,
this natural map decomposes as $\phi_r\oplus\id_{B_{r-2,r}}$. This proves \eqref{kerphir}.
\end{proof}

This readily gives a coarse upper bound on the $\beta_{r-1,r}$:
\begin{lemma}\label{grossier}
We have $\beta_{1,2}\leq\frac{k(k-1)}{2}$, and $\beta_{r-1,r}\leq(k-1)\beta_{r-2,r-1}$ for $r\geq3$.
Hence $\beta_{r-1,r}(\code)\leq\frac{k}{2}(k-1)^{r-1}$ for any $r\geq2$.
\end{lemma}
\begin{proof}
We have $\dim(\code\deux)\geq\dim(C)=k$ hence $\beta_{1,2}=\binom{k+1}{2}-\dim(\code\deux)\leq\frac{k(k-1)}{2}$.
Now let $r\geq3$. As $M_{r-2}$ is a submodule of the free module $F_{r-3}$, multiplication by $X_1$ is injective
on $M_{r-2}$. Hence $M_{r-2,r-1}\simeq X_1M_{r-2,r-1}\subset\im(\phi_r)$
from which it follows $\beta_{r-1,r}=k\beta_{r-2,r-1}-\dim\im(\phi_r)\leq(k-1)\beta_{r-2,r-1}$.
\end{proof}

\begin{proposition}
The $M_{r-1,r}$ can be computed iteratively as follows.
For $r=3$:
\beq\label{M23_noyau}
M_{2,3}=\ker(\;I_2(\code)\tens S_1\longto S_3\;)
\eeq
where $I_2(\code)\tens S_1\longto S_3$ is the natural multiplication map. Then for $r\geq 4$:
\beq\label{Mii+1_noyau_itere}
M_{r-1,r}=\ker(\;M_{r-2,r-1}\tens S_1\longto M_{r-3,r-2}\tens S_2\;)
\eeq
where the map $\psi_r:M_{r-2,r-1}\tens S_1\longto M_{r-3,r-2}\tens S_2$ is obtained first by tensoring
the inclusion $M_{r-2,r-1}\subset M_{r-3,r-2}\tens S_1$ by $S_1$,
and then composing with the multiplication map $S_1\tens S_1\longto S_2$.
\end{proposition}
\begin{proof}
First, \eqref{M23_noyau} is just the case $r=3$ of \eqref{kerphir}, composed with the inclusion $I_3(\code)\subset S_3$.
Likewise to prove \eqref{Mii+1_noyau_itere} we must show $\ker(\phi_r)=\ker(\psi_r)$ for $r\geq4$.
For this we just observe that the diagram 
\beq
\begin{tikzcd}[row sep=tiny]
& M_{r-2,r}\arrow[rd,start anchor=east,end anchor={[yshift=4pt]west}] & \\
M_{r-2,r-1}\tens S_1\arrow[ru,"\phi_r",start anchor={[yshift=4pt]east},end anchor=west]\arrow[rd,"\psi_r",start anchor={[yshift=-4pt]east},end anchor={[yshift=4pt]west}] & & F_{r-3,r}\\
& M_{r-3,r-2}\tens S_2=F_{r-3,r-2}\tens S_2\arrow[ru,start anchor={[yshift=4pt]east},end anchor={[yshift=-4pt]west}] &
\end{tikzcd}
\eeq
commutes, with the two arrows on the right injective.
\end{proof}

Thus we have two descriptions of $M_{r-1,r}$.
The description $M_{r-1,r}=\ker(\phi_r)$ is closer to the abstract definition,
and will be used later to estimate the value of the Betti numbers.
The description $M_{r-1,r}=\ker(\psi_r)$ is more amenable to effective computation. 

\begin{corollary}\label{I2_determine_tout}
The $M_{r-1,r}(\code)$ only depend on $I_2(\code)$.
\end{corollary}

To make things entirely explicit we now introduce the following matrices.
Assume $\code$ is given by a generator matrix $\mat{G}\in\F^{k\times n}$, with rows $\mot{c}_1,\dots,\mot{c}_k$.
Choose monomial bases $\mathcal{M}_1=(X_a)_{1\leq a\leq k}$, $\mathcal{M}_2=(X_aX_b)_{1\leq a\leq b\leq k}$,
and $\mathcal{M}_3=(X_aX_bX_c)_{1\leq a\leq b\leq c\leq k}$ of $S_1$, $S_2$, $S_3$,
ordered with respect to some monomial order.

\begin{definition}\label{G2}
The squared matrix of $\mat{G}$ is the matrix
\beq
\mat{M}_2\in\F^{\,\binom{k+1}{2}\,\times\, n}
\eeq
with rows indexed by $\mathcal{M}_2$: the row corresponding to $X_aX_b$ is $\mot{c}_a\mot{c}_b$.
\end{definition}
Then $I_2(\code)$ is the left kernel of $\mat{M}_2$, and we choose a basis $\mathcal{B}_2$ of this space.
Thus $\mathcal{B}_2$ consists of $\beta_{1,2}$ vectors, each of which has its entries indexed by $\mathcal{M}_2$.

\begin{definition}\label{Mac3}
The degree $3$ Macaulay matrix of $\mathcal{B}_2$ is the matrix
\beq
\mat{M}_3\in\F^{\:k\beta_{1,2}\,\times\,\binom{k+2}{3}}
\eeq
with rows indexed by $\mathcal{M}_1\times\mathcal{B}_2$
and columns indexed by $\mathcal{M}_3$:
the row corresponding to $(X_a,\mot{q})\in\mathcal{M}_1\times\mathcal{B}_2$,
where $\mot{q}=(q_M)_{M\in\mathcal{M}_2}$, has entry $q_M$ at position corresponding to $X_aM\in\mathcal{M}_3$, and $0$ elsewhere.
\end{definition}
(Said otherwise, the rows of $\mat{M}_3$ are formed by multiplying all
$\mot{q}\in\mathcal{B}_2$, seen as quadratic forms $\mot{q}=\sum_{1\leq b\leq c\leq k}q_{X_bX_c}X_bX_c$, by all variables $X_a$.)

Then by \eqref{M23_noyau}, $M_{2,3}$ is the left kernel of $\mat{M}_3$, and we choose a basis $\mathcal{B}_3$ of this space.
Thus $\mathcal{B}_3$ consists of $\beta_{2,3}$ vectors, each of which has its entries indexed by $\mathcal{M}_1\times\mathcal{B}_2$.

Now let $r\geq4$, and assume inductively for all $3\leq i\leq r-1$ we have constructed a basis $\mathcal{B}_i$ of $M_{i-1,i}$,
consisting of $\beta_{i-1,i}$ vectors, each of which has its entries indexed by $\mathcal{M}_1\times\mathcal{B}_{i-1}$.
\begin{definition}\label{Mac+}
The degree $r$ blockwise Macaulay matrix of $\mathcal{B}_{r-1},\mathcal{B}_{r-2}$ is the matrix
\beq
\mat{M}_{r}\in\F^{\:k\beta_{r-2,r-1}\,\times\,\binom{k+1}{2}\beta_{r-3,r-2}}
\eeq
with rows indexed by $\mathcal{M}_1\times\mathcal{B}_{r-1}$
and columns indexed by $\mathcal{M}_2\times\mathcal{B}_{r-2}$:
the row corresponding to $(X_a,\mot{s})\in\mathcal{M}_1\times\mathcal{B}_{r-1}$,
where $\mot{s}=(s_{X_b,\mot{t}})_{X_b\in\mathcal{M}_1,\mot{t}\in\mathcal{B}_{r-2}}$,
has entry $s_{X_b,\mot{t}}$ at position corresponding to $(X_aX_b,\mot{t})\in\mathcal{M}_2\times\mathcal{B}_{r-2}$, and $0$ elsewhere.
\end{definition}
(Said otherwise, the rows of $\mat{M}_r$ are formed by multiplying all
$\mot{s}\in\mathcal{B}_{r-1}$ by all variables $X_a$,
where each $\mot{s}$ is seen as a formal sum $\mot{s}=\sum_{\mot{t}\in\mathcal{B}_{r-2}}L_{\mot{t}}[\mot{t}]$
whose coefficients are linear forms $L_{\mot{t}}=\sum_{1\leq b\leq k}s_{X_b,\mot{t}}X_b$, 
resulting in formal sums $X_a\mot{s}=\sum_{\mot{t}\in\mathcal{B}_{r-2}}(X_aL_{\mot{t}})[\mot{t}]$ whose coefficients are quadratic forms.)

Then by \eqref{Mii+1_noyau_itere}, $M_{r-1,r}$ is the left kernel of $\mat{M}_{r}$, and we proceed.

All this is summarized in Algorithm~\ref{algobases}.
\begin{algorithm}
\caption{Compute bases $\mathcal{B}_r$ of $M_{r-1,r}(\code)$ up to some degree $D$}\label{algobases}
\textbf{Input:}
\begin{minipage}[t]{.75\textwidth}
\begin{itemize}
\item a generator matrix $\mat{G}$ of the $[n,k]$-code $\code$
\item a target degree $D\leq k$
\end{itemize}
\end{minipage}
\begin{algorithmic}[1]
\State \textbf{construct} the matrix $\mat{M}_2$ according to Definition~\ref{G2}
\State \textbf{compute} a left kernel basis $\mathcal{B}_2$ of $\mat{M}_2$
\State \textbf{construct} the matrix $\mat{M}_3(\mathcal{B}_2)$ according to Definition~\ref{Mac3}
\State \textbf{compute} a left kernel basis $\mathcal{B}_3$ of $\mat{M}_3$
\For{$r=4\dots D$}
\State \textbf{construct} the matrix $\mat{M}_{r}(\mathcal{B}_{r-1},\mathcal{B}_{r-2})$ according to Definition~\ref{Mac+}
\State \textbf{compute} a left kernel basis $\mathcal{B}_{r}$ of $\mat{M}_{r}$
\EndFor
\end{algorithmic}
\end{algorithm}

An implementation can be found in \cite{githubsyzygies}, which was used to compute all the examples in this text.

Observe that this algorithm only relies on mere linear algebra, and does not make use of any Gröbner basis theory.
However, the two topics are clearly related.
In particular, the many linear algebra optimizations used in Gröbner basis algorithms certainly apply here also.
This will be discussed later in the text, when we'll deal with complexity estimates.


\subsection*{Further properties}

We will have a particular interest in the length of the linear strand, or equivalently, in the following quantity:
\begin{definition}
For a linear code $\code$, we set
\beq
\rmax(\code)=\max\{r:\:\beta_{r-1,r}(\code)>0\}.
\eeq
\end{definition}
As the minimal resolution of a $[n,k]$-code $\code$ has length $k-1$, we always have $\rmax(\code)\leq k$.

Sometimes we only control the syzygies of certain subideals of $I(\code)$, and from
these we want to deduce information on the syzygies of the whole ideal.
Intuitively, we expect that syzygies between elements of the subideals could be seen as syzygies between elements of the ideal.
Of course this fails in general, but it remains true if we restrict to the linear strand:

\begin{proposition}\label{directsum}
Let $M$ be a finite $S$-module generated in degrees $\geq d_0$, and
suppose its degree $d_0$ part $M_{d_0}$ decomposes as a direct sum
\beq
M_{d_0}=A_{d_0}\oplus B_{d_0}.
\eeq
Let $A=\linspan[S]{A_{d_0}}$ and $B=\linspan[S]{B_{d_0}}$ be the sub-$S$-modules of $M$ generated by $A_{d_0}$ and $B_{d_0}$.
Let $\widehat{M}$ be the (first) syzygy module of $M$, and $\widehat{A}$ and $\widehat{B}$ those of $A$ and $B$.
Then we have a natural inclusion
\beq
\widehat{A}_{d_0+1}\oplus \widehat{B}_{d_0+1}\;\subset\;\widehat{M}_{d_0+1}
\eeq
the cokernel of which identifies with $A_{d_0+1}\cap B_{d_0+1}$, hence a (non-canonical) isomorphism
\beq
\widehat{M}_{d_0+1}\;\simeq\;\widehat{A}_{d_0+1}\oplus \widehat{B}_{d_0+1}\oplus (A_{d_0+1}\cap B_{d_0+1}).
\eeq
\end{proposition}
\begin{proof}
We have surjective maps
\beq
(A_{d_0}\oplus B_{d_0})\tens S_1\;\surj\; A_{d_0+1}\oplus B_{d_0+1}\;\surj\; A_{d_0+1}+ B_{d_0+1}
\eeq
with $\widehat{A}_{d_0+1}\oplus \widehat{B}_{d_0+1}$ the kernel of the leftmost map,
$A_{d_0+1}\cap B_{d_0+1}$ the kernel of the rightmost map,
and $\widehat{M}_{d_0+1}$ the kernel of the composite map.
We conclude with the associated kernel-cokernel exact sequence.
\end{proof}
\begin{corollary}\label{sous-resolutions}
Suppose $I_2(\code)$ contains a certain direct sum of $l$ subspaces:
\beq
I_2(\code) \supset V^{(1)}\oplus\cdots\oplus V^{(l)}.
\eeq
Then the minimal resolution of $\code\deux[\cdot]$ canonically admits the direct sum
of the linear strands of the minimal resolutions of $S/\linspan{V^{(1)}}$, ..., $S/\linspan{V^{(l)}}$
as a subcomplex.
In particular for all $r\geq2$ we have
\beq
\beta_{r-1,r}(\code)\;\geq\;\beta_{r-1,r}(S/\linspan{V^{(1)}})+\cdots+\beta_{r-1,r}(S/\linspan{V^{(l)}}).
\eeq
\end{corollary}
\begin{proof}
Apply Proposition~\ref{directsum} repeatedly.
\end{proof}

Recall that $\pi\code$ is a punctured code of $\code$ if it is obtained from $\code$ by discarding some given set of coordinates
(equivalently, discarding some subset of the associated set of points in projective space).
\begin{corollary}\label{puncture}
Let $\pi\code$ be a punctured code of $\code$. Assume $\dim(\pi\code)=\dim(\code)$.
Then for all $r\geq2$ we have
\beq
\beta_{r-1,r}(\pi\code)\geq\beta_{r-1,r}(\code)
\eeq
hence
\beq
\rmax(\pi\code)\geq \rmax(\code).
\eeq
\end{corollary}
\begin{proof}
Let $S=\F[X_1,\dots,X_k]$ where $k=\dim(\pi\code)=\dim(\code)$.
Then the surjective map $S\longto(\pi\code)\deux[\cdot]$ factors through $\code\deux[\cdot]$,
so that $I_2(\pi\code)\supset I_2(\code)$.
\end{proof}

Last, the following easy result is also useful:
\begin{lemma}\label{extscal}
Minimal resolutions are preserved under extension of scalars. In particular if a code $\code$
is defined over $\F$, and if $\F\subset\K$ is a field extension, then $\beta_{i,j}(\code)=\beta_{i,j}(\code_{\K})$ for all $i,j$.
\end{lemma}

\section{Lower bounds from the Eagon-Northcott complex}\label{secEN}

\subsection*{Generalities}

Let $R$ be a ring, and for integers $f\geq g$, temporarily switching to column vector convention,
let $\mat{\Phi}\in R^{g\times f}$ define a linear map
\beq
\mat{\Phi}:F=R^f\longto G=R^g.
\eeq
The Eagon-Northcott complex \cite{EagNor62} of $\mat{\Phi}$ is the following complex of $R$-modules,
defined in terms of the exterior and (dual of) symmetric powers of $F$ and $G$:
\beq\label{EN}
0\to(\Sym^{f-g}G)^\vee\tens\bigwedge^fF\to\cdots\to G^\vee\tens\bigwedge^{g+1}F\to\bigwedge^gF\overset{\bigwedge^g\mat{\Phi}}\longto\bigwedge^gG\simeq R.
\eeq
Under mild hypotheses, this complex is exact \cite[Th.~A.2.60 \& Th.~6.4]{Eisenbud05},
so it defines a resolution of the quotient $R/\im(\bigwedge^g\mat{\Phi})$
defined by the ideal generated by the maximal minors of $\mat{\Phi}$.

In this work we will only need the case $g=2$. In this case the complex has length $f-1$,
and for $2\leq r\leq f$ its $(r-1)$-th module is free of rank
\beq\label{rkEN}
\rk\left((\Sym^{r-2}G)^\vee\tens\bigwedge^{r}F\right) \;=\; (r-1)\binom{f}{r}.
\eeq
All this can be made explicit, in coordinates. 
Let
\beq
\mat{\Phi}=\left(\begin{array}{cccc}x_1 & x_2 & \dots & x_f\\ x'_1 & x'_2 & \dots & x'_f\end{array}\right)
\eeq
for $x_i,x'_i\in R$. Then:
\begin{itemize}
\item $\im(\bigwedge^2\mat{\Phi})$ is generated by the $\binom{f}{2}$ minors
\beq
q_{i,j}=x_ix'_j-x_jx'_i
\eeq
for $1\leq i<j\leq f$,
\item these $q_{ij}$ are annihilated by the $2\binom{f}{3}$ relations
\beq
\begin{split}
r_{ijk}&=x_iq_{jk}-x_jq_{ik}+x_kq_{ij}\\
r'_{ijk}&=x'_iq_{jk}-x'_jq_{ik}+x'_kq_{ij}
\end{split}
\eeq
for $1\leq i<j<k\leq f$,
\item these $r_{ijk}$ and $r'_{ijk}$ are annihilated by the $3\binom{f}{4}$ relations
\beq
\begin{split}
s_{ijkl}&=x_ir_{jkl}-x_jr_{ikl}+x_kr_{ijl}-x_lr_{ijk}\\
s'_{ijkl}&=x_ir'_{jkl}-x_jr'_{ikl}+x_kr'_{ijl}-x_lr'_{ijk}+\\
&\phantom{=x_ir'_{jkl}-x_j}+x'_ir_{jkl}-x'_jr_{ikl}+x'_kr_{ijl}-x'_lr_{ijk}\\
s''_{ijkl}&=x'_ir'_{jkl}-x'_jr'_{ikl}+x_kr'_{ijl}-x_lr'_{ijk}
\end{split}
\eeq
for $1\leq i<j<k<l\leq f$,
\end{itemize}
and so on. These are just the cases $r=2,3,4$ of the more general:
\begin{proposition}\label{explicitEN}
In the multivariate polynomial ring
\beq
R[Z_{r,j;i_1,\dots,i_r}]
\eeq
where, for each $2\leq r\leq f$, indices range over the $(r-1)\binom{f}{r}$ values 
$1\leq j\leq r-1$ and $1\leq i_1<\dots<i_r\leq f$,
construct polynomials $s_{r,j;i_1,\dots,i_r}$ as follows. First for $r=2$ we define the constant polynomials
\beq\label{s2}
s_{2,1;i_1,i_2}=x_{i_1}x'_{i_2}-x'_{i_1}x_{i_2}
\eeq
given by the minors of $\mat{\Phi}$; and for $r\geq3$ we define linear polynomials
\beq\label{sr}
s_{r,j;i_1,\dots,i_r}=\sum_{u=1}^r(-1)^{u-1}(x_{i_u}Z_{r-1,j;i_1,\dots,\widehat{i_u},\dots,i_r}+x'_{i_u}Z_{r-1,j-1;i_1,\dots,\widehat{i_u},\dots,i_r})
\eeq
where notation $\widehat{i_u}$ means $i_u$ is omitted, and we replace $Z_{r,j;i_1,\dots,i_r}$ with zero if $j\leq0$ or $j\geq r$.
Then, for all $r\geq3$, we have
\beq
\mot{s}_r(\mot{s}_{r-1})=\mot{0}.
\eeq
\end{proposition}
\begin{proof}
Either express \eqref{EN} in the standard bases of $F=R^f$ and $G=R^2$.
Or more directly, consider the $\mot{Z}_r$ and $\mot{s}_r$ as $(r-1)$-tuples of (exterior) $r$-vectors,
observe that \eqref{s2} means $\mot{s}_{2,1}=\mot{x}\wedge\mot{x}'$,
that \eqref{sr} means $\mot{s}_{r,j}=\mot{x}\wedge\mot{Z}_{r-1,j}+\mot{x}'\wedge\mot{Z}_{r-1,j-1}$,
so the alternating property gives
$\mot{s}_{3,1}(\mot{s}_2)=\mot{x}\wedge\mot{x}\wedge\mot{x}'=\mot{0}$,
$\;\mot{s}_{3,2}(\mot{s}_2)=\mot{x}'\wedge\mot{x}\wedge\mot{x}'=\mot{0}$,
and for $r\geq4$:
\beq
\begin{split}
\mot{s}_{r,j}(\mot{s}_{r-1})&=\mot{x}\wedge\mot{x}\wedge\mot{Z}_{r-2,j}+\mot{x}\wedge\mot{x}'\wedge\mot{Z}_{r-2,j-1}+\\
&\phantom{=\mot{x}\wedge\mot{x}\wedge\mot{Z}}+\mot{x}'\wedge\mot{x}\wedge\mot{Z}_{r-2,j-1}+\mot{x}'\wedge\mot{x}'\wedge\mot{Z}_{r-2,j-2}=\mot{0}.
\end{split}
\eeq
\end{proof}

\subsection*{Shortening (or projection from a point)}

We have interest in how this Eagon-Northcott complex interacts with reduction to a subcode,
and in particular with shortening. By induction it suffices to consider the case of
a codimension $1$ subcode, or a $1$-shortening respectively.

So let $\code$ be a $[n,k]$-code over $\F$, and $\code_{\mathcal{H}}$ a codimension $1$ subcode.
We will assume $\code$ projective, but then $\code_{\mathcal{H}}$ need not be so.
Choose a basis $\mot{c}_1,\dots,\mot{c}_{k-1}$ of $\code_{\mathcal{H}}$,
and complete it to a basis $\mot{c}_1,\dots,\mot{c}_k$ of $\code$.
Let $\mat{G}_{\mathcal{H}}$ and $\mat{G}$ be the corresponding generating matrices of $\code_{\mathcal{H}}$ and $\code$.
Let also $S=\F[X_1,\dots,X_k]$ be the polynomial ring in $k$ indeterminates,
and $S_{\mathcal{H}}=\F[X_1,\dots,X_{k-1}]$ its subring in the first $k-1$ indeterminates.
We have a commutative diagram
\beq\label{diagramme_sous-code}
\begin{CD}
S_{\mathcal{H}} @>>> \code_{\mathcal{H}}\deux[\cdot]\\
@VVV @VVV\\
S @>>> \code\deux[\cdot]
\end{CD}
\eeq
where horizontal maps denote evaluation, and vertical maps inclusion. Also set
\beq
\mathcal{H}=S_{\mathcal{H},1}=\linspan[\F]{X_1,\dots,X_{k-1}}\subset S_1
\eeq
and let $\mot{p}_{\mathcal{H}}=(0:\dots:0:1)^\top\in\PP^{k-1}$ be the associated point.
Last, let $\{\mot{p}_1,\dots,\mot{p}_n\}\subset\PP^{k-1}$ be the set of points defined by the columns of $\mat{G}$.
Then we have two possibilities:
\begin{itemize}
\item Either $\mot{c}_1,\dots,\mot{c}_{k-1}$ all vanish at some common position $i$,
or equivalently, $\mot{p}_{\mathcal{H}}=\mot{p}_i$.
Puncturing this position, we can identify $\code_{\mathcal{H}}$ with the $1$-shortened subcode of $\code$ at $i$.
Then the set of points in $\PP^{k-2}$ defined by the columns of $\mat{G}_{\mathcal{H}}$
is the image of $\{\mot{p}_1,\dots,\widehat{\mot{p}}_i,\dots,\mot{p}_n\}$ under the projection from $\mot{p}_i$.
\item Otherwise, if there is no such $i$, then $\code_{\mathcal{H}}$ is a ``general'' codimension $1$ subcode,
and the set of points in $\PP^{k-2}$ defined by the columns of $\mat{G}_{\mathcal{H}}$
is the image of $\{\mot{p}_1,\dots,\mot{p}_n\}$ under the projection from $\mot{p}_{\mathcal{H}}$.
\end{itemize}
(As already observed, $\code_{\mathcal{H}}$ need not be projective, so these image points need not be distinct.)

Let $f\geq g\geq d\geq 2$ be integers,
and let $\mat{\Phi}\in S_1^{g\times f}$ be a matrix whose entries are homogeneous linear forms, i.e. elements of $S_1$.
Recall from \cite[\S6B]{Eisenbud05} that such a matrix is said to be $1$-generic if
for any nonzero $\mot{a}\in\overline{\F}^g$ and $\mot{b}\in\overline{\F}^f$, $\mot{a}\mat{\Phi}\mot{b}^\top$ is nonzero.
Let $V_{\mat{\Phi}}\subset S_1^g$ be the column span of $\mat{\Phi}$, so $\dim(V_{\mat{\Phi}})\leq f$, and set
\beq
V_{\mat{\Phi},\mathcal{H}}=V_{\mat{\Phi}}\cap\mathcal{H}^g.
\eeq
Set $f_{\mathcal{H}}=\dim(V_{\mat{\Phi},\mathcal{H}})\geq\dim(V_{\mat{\Phi}})-g$, and then let $\mat{\Phi}_{\mathcal{H}}\in\mathcal{H}^{g\times f_{\mathcal{H}}}$
be a matrix whose columns form a basis of $V_{\mat{\Phi},\mathcal{H}}$.
\begin{proposition}\label{ENrac}
Under these hypotheses:
\begin{enumerate}
\item If $I_d(\code)$ contains the $d\times d$ minors of $\mat{\Phi}$,
then $I_d(\code_{\mathcal{H}})$ contains the $d\times d$ minors of $\mat{\Phi}_{\mathcal{H}}$.
\item If $\mat{\Phi}$ is $1$-generic, then $\mat{\Phi}_{\mathcal{H}}$ is $1$-generic, and $f_{\mathcal{H}}\geq f-g$.
\item If $I_d(\code)$ contains the $d\times d$ minors of $\mat{\Phi}$, if $\mat{\Phi}$ is $1$-generic, 
and if $\code_{\mathcal{H}}$ is a $1$-shortened subcode of $\code$, then $f_{\mathcal{H}}\geq f-d+1$.
\end{enumerate}
\end{proposition}
\begin{proof}
By \eqref{diagramme_sous-code} we have $I_d(\code_{\mathcal{H}})=I_d(\code)\cap S_{\mathcal{H}}$.
Columns of $\mat{\Phi}_{\mathcal{H}}$ belong to the column span of $\mat{\Phi}$,
so $d\times d$ minors of $\mat{\Phi}_{\mathcal{H}}$ are linear combinations of $d\times d$ minors of $\mat{\Phi}$.
On the other hand, $\mat{\Phi}_{\mathcal{H}}$ has coefficients in $\mathcal{H}$, so its $d\times d$ minors
belong to $S_{\mathcal{H}}$. This proves~1.

Assume $\mat{\Phi}$ is $1$-generic.
As columns of $\mat{\Phi}_{\mathcal{H}}$ are linearly independent and belong to the column space of $\mat{\Phi}$,
we can write $\mat{\Phi}_{\mathcal{H}}=\mat{\Phi}\mat{M}$ where $\mat{M}\in\F^{f\times f_{\mathcal{H}}}$
has $\rk(\mat{M})=f_{\mathcal{H}}$. Then for $\mot{a}\in\overline{\F}^g$ and $\mot{b}\in\overline{\F}^{f_{\mathcal{H}}}$ nonzero
we have $\mot{a}\mat{\Phi}_{\mathcal{H}}\mot{b}^\top=\mot{a}\mat{\Phi}\mat{M}\mot{b}^\top\neq0$ because $\mat{M}\mot{b}^\top$
is nonzero. This proves that $\mat{\Phi}_{\mathcal{H}}$ is $1$-generic.
Then it is easily seen that a $1$-generic matrix has linearly independent columns, so $\dim(V_{\mat{\Phi}})=f$, hence $f_{\mathcal{H}}\geq f-g$.
This finishes the proof of~2.

Now we prove 3. Assume $\code_{\mathcal{H}}$ is the $1$-shortened subcode of $\code$ at position~$i$,
hence $\mot{p}_i=\mot{p}_{\mathcal{H}}=(0:\dots:0:1)^\top$.
As $\mat{\Phi}$ is $1$-generic, we have $\dim(V_{\mat{\Phi}})=f$.
Let $\ev_{\mot{p}_{\mathcal{H}}}:S\longto\F$ denote evaluation at $\mot{p}_{\mathcal{H}}$.
Applied coordinatewise, we also have evaluation maps $\ev_{\mot{p}_{\mathcal{H}}}:S^g\longto\F^g$
and $\ev_{\mot{p}_{\mathcal{H}}}:S^{g\times f}\longto\F^{g\times f}$.
Also we can restrict $\ev_{\mot{p}_{\mathcal{H}}}$ to subspaces,
so for instance $\mathcal{H}=\ker(\ev_{\mot{p}_{\mathcal{H}}}:S_1\longto \F)$.
It then follows
\beq\label{kerevpH}
V_{\mat{\Phi},\mathcal{H}}=V_{\mat{\Phi}}\cap\mathcal{H}^g=\ker(\,\ev_{\mot{p}_{\mathcal{H}}}:V_{\mat{\Phi}}\longto \F^g\,),
\eeq
while the image $\ev_{\mot{p}_{\mathcal{H}}}(V_{\mat{\Phi}})\subset\F^g$ is
the column span of $\ev_{\mot{p}_{\mathcal{H}}}(\mat{\Phi})\in\F^{g\times f}$.
As the $d\times d$ minors of $\mat{\Phi}$ belong to $I_d(\code)$, they vanish at each column of $\mat{G}$,
in particular they vanish at $\mot{p}_{\mathcal{H}}=\mot{p}_i$. But this means precisely that the $d\times d$ minors
of $\ev_{\mot{p}_{\mathcal{H}}}(\mat{\Phi})$ all vanish, or equivalently,
\beq
\dim(\ev_{\mot{p}_{\mathcal{H}}}(V_{\mat{\Phi}}))\leq d-1.
\eeq
Joint with \eqref{kerevpH} this gives $f_{\mathcal{H}}=\dim(V_{\mat{\Phi},\mathcal{H}})\geq \dim(V_{\mat{\Phi}})-(d-1)=f-d+1$.
\end{proof}

\subsection*{Application to alternant codes}

Consider a code $\code\in\Alt^\perp_{q,m,n,t}$, assumed to be proper.
By \eqref{decompose_alternant}, after extension of scalars, $\code_{\F_{q^m}}$ admits as a basis
the $k=mt$ vectors $(\mot{y}\mot{x}^a)^{q^u}$ for $0\leq a\leq t-1$ and $0\leq u\leq m-1$.
Rename the $k=mt$ variables in our polynomial ring $S$ accordingly, so that our evaluation map now is
\beq
\label{evalalt}
S=\F_{q^m}[X^{(u)}_a]\longto \code_{\F_{q^m}}\deux[\cdot]
\eeq
where $X^{(u)}_a$ evaluates as $(\mot{y}\mot{x}^a)^{q^u}$.

If $\mat{M}$ is a matrix (or more generally an expression) in the $X^{(u)}_a$,
we denote by $\mat{M}^{(v)}$ the same matrix (or expression) with each $X^{(u)}_a$ replaced by $X^{(u+v)}_a$,
where $u+v$ is considered mod $m$. We also write $\mat{M}'=\mat{M}^{(1)}$, $\mat{M}''=\mat{M}^{(2)}$, etc.

Set
\beq
e=\lfloor\log_q(t-1)\rfloor,
\eeq
and for any $0\leq u\leq e$ define a $2\times(t-q^u)$ matrix
\beq
\mat{B}_u=
\left(\begin{array}{cccc}
X^{(0)}_0 & X^{(0)}_1 & \dots & X^{(0)}_{t-1-q^u}\\
X^{(0)}_{q^u} & X^{(0)}_{q^u+1} & \dots & X^{(0)}_{t-1}
\end{array}\right).
\eeq
We then define the block matrix
\beq\label{defPhi}
\mat{\Phi}=
\left(\begin{array}{c|c|c|c}
\mat{B}^{(e)}_0 & \mat{B}^{(e-1)}_1 & \cdots & \mat{B}^{(0)}_e
\end{array}\right)
\eeq
of total size $2\times f$ where $f=(e+1)t-\frac{q^{e+1}-1}{q-1}$.
Observe that this matrix only depends on $q$ and $t$, debatably on $m$,
but certainly not on $n$ nor on the specific choice of $\code$.

\begin{lemma}\label{minPhiI2C}
The $2\times 2$ minors of $\mat{\Phi}$ belong to $I_2(\code_{\F_{q^m}})$.
\end{lemma}
\begin{proof}
The $2\times 2$ minor defined by the columns
$\left[\begin{array}{c}X^{(e-u)}_a\\X^{(e-u)}_{a+q^u}\end{array}\right]$ of $\mat{B}^{(e-u)}_u$ and
$\left[\begin{array}{c}X^{(e-v)}_b\\X^{(e-v)}_{b+q^v}\end{array}\right]$ of $\mat{B}^{(e-v)}_v$ is
\beq
X^{(e-u)}_aX^{(e-v)}_{b+q^v}-X^{(e-v)}_bX^{(e-u)}_{a+q^u}
\eeq
which evaluates under \eqref{evalalt} to
\beq
(\mot{y}\mot{x}^a)^{q^{e-u}}(\mot{y}\mot{x}^{b+q^v})^{q^{e-v}}-(\mot{y}\mot{x}^b)^{q^{e-v}}(\mot{y}\mot{x}^{a+q^u})^{q^{e-u}}=\mot{0}.
\eeq
\end{proof}

\begin{lemma}\label{1gen}
The matrix $\mat{\Phi}$ is $1$-generic.
\end{lemma}
\begin{proof}
Let $\mot{a}\in\overline{\F}_q^2$ and $\mot{b}\in\overline{\F}_q^f$ be nonzero.
Let $a_i$ be the rightmost nonzero entry of $\mot{a}$ and $b_j$ the rightmost nonzero entry of $\mot{b}$.
Then the $X_c^{(u)}$ variable that multiplies $a_ib_j$ in $\mot{a}\mat{\Phi}\mot{b}^\top$ does not appear
elsewhere in $\mot{a}\mat{\Phi}\mot{b}^\top$, so $\mot{a}\mat{\Phi}\mot{b}^\top\neq0$.
\end{proof}

\begin{theorem}\label{thlowerboundalt}
Let $\code\in\Alt^\perp_{q,m,n,t}$ be proper, $e=\lfloor\log_q(t-1)\rfloor$, $f=(e+1)t-\frac{q^{e+1}-1}{q-1}$.
For any $s\geq 0$, let $\code_s$ be a $s$-shortened subcode of $\code$.
Then for all $r\geq 2$ we have $\beta_{r-1,r}(\code_s)\geq(r-1)\binom{f-s}{r}$,
hence
\beq\label{lowerboundrmaxalt}
\rmax(\code_s)\geq f-s.
\eeq
\end{theorem}
\begin{proof}
By Lemma~\ref{extscal} we can extend scalars to $\F_{q^m}$.
By Lemma~\ref{minPhiI2C}, $I_2(\code)$ then contains the $2\times 2$ minors of
the $2\times f$ matrix $\mat{\Phi}$, 
and this matrix $\mat{\Phi}$ is $1$-generic by Lemma~\ref{1gen}.
By Proposition~\ref{ENrac} applied $s$ times (with $g=d=2$),
$I_2(\code_s)$ then contains the $2\times 2$ minors of a $2\times f_s$ matrix $\mat{\Phi}_s$, where $f_s\geq f-s$,
and this matrix $\mat{\Phi}_s$ is $1$-generic also.
Then by~\cite[\S6B]{Eisenbud05}, the $2\times 2$ minors of $\mat{\Phi}_s$ are linearly independent,
and its Eagon-Northcott complex is exact with only nonzero Betti numbers $\beta_{0,0}=1$
and $\beta_{r-1,r}=(r-1)\binom{f_s}{r}$ for $2\leq r\leq f_s$.
We then conclude with Corollary~\ref{sous-resolutions} (with $l=1$ and $V^{(1)}=\im(\bigwedge^2\mat{\Phi}_s)$).
\end{proof}

\subsection*{Further results on syzygies of dual alternant codes}

It is possible to improve the lower bound on $\beta_{r-1,r}(\code_s)$, using not only $\mat{\Phi}=\mat{\Phi}^{(0)}$, but also
its conjugates $\mat{\Phi}^{(1)},\dots,\mat{\Phi}^{(m-1)}$.
Unfortunately we only have very partial results in this direction.
These partial results will not be used in the remainder of the text, but we present them nevertheless in the hope that
they could motivate further study.

Let us first consider the simpler case $s=0$, i.e. without shortening.
Let then $\code$ be a proper dual alternant code as in Theorem~\ref{thlowerboundalt}.

The Eagon-Northcott complex of each conjugate $\mat{\Phi}^{(u)}$ is a subcomplex of the minimal resolution of $\code$,
so that in each degree $r$ we get a subspace
\beq\label{defVr}
V_r=M_{r-1,r}(\mat{\Phi}^{(0)})+\cdots+M_{r-1,r}(\mat{\Phi}^{(m-1)})\subset M_{r-1,r}(\code),
\eeq
hence $\beta_{r-1,r}(\code)\geq\dim V_r$. Two questions then arise naturally:
\begin{itemize}
\item[(Q1)] Is the sum in \eqref{defVr} direct? If not, what are the linear dependencies between the $M_{r-1,r}(\mat{\Phi}^{(u)})$ in $M_{r-1,r}(\code)$?
\item[(Q2)] Is the inclusion $V_r\subset M_{r-1,r}(\code)$ an equality, i.e. do the $M_{r-1,r}(\mat{\Phi}^{(u)})$ together generate $M_{r-1,r}(\code)$?
\end{itemize}
We can answer these questions in degree $r=2$.
\begin{definition}\label{defI2uv}
For any $u,v$ we let
\beq
I_2^{(u,v)}\subset I_2(\code)
\eeq
be the subspace generated by the quadratic forms
\beq\label{qabcd}
q^{(u,v)}_{a,b,c,d}=X_a^{(u)}X_d^{(v)}-X_b^{(v)}X_c^{(u)}
\eeq
for all $a,b,c,d$ such that
\beq\label{condabcd}
aq^u+dq^v=bq^v+cq^u \mod q^m-1
\eeq
(so that these quadratic forms indeed evaluate to $\mot{0}$).
\end{definition}
Observe that $I_2^{(u,v)}=I_2^{(v,u)}$. Thus, in order to avoid these $I_2^{(u,v)}$ appearing twice, we will impose some conditions
on the pairs $u,v$ we consider. Say, we consider only those where $v=u-w \mod m$ with $0\leq w<m/2$,
and also in case $m$ even those where $v=u-m/2 \mod m$ with $0\leq u<m/2$.
\begin{lemma}\label{I2uvindep}
With such a restriction, the $I_2^{(u,v)}$ are in direct sum in $I_2(\code)$.
\end{lemma}
\begin{proof}
Consider a linear combination $\sum_{u,v}q^{(u,v)}$ for general $q^{(u,v)}\in I_2^{(u,v)}$.
If some $q^{(u,v)}$ is nonzero, then it contains a monomial $X_i^{(u)}X_j^{(v)}$ for some $i,j$.
However this monomial does not appear in the $q^{(u',v')}$ with $(u',v')\neq(u,v)$,
so the linear combination is nonzero.
\end{proof}
\begin{lemma}\label{baseI2uv}
If $v=u-w \mod m$ with $1\leq w\leq m/2$, then a basis of $I_2^{(u,v)}$ is given by those $2\times 2$ minors of the matrix
\begin{equation*}
\left(\begin{array}{c|c}
\mat{B}^{(u)}_0 & \mat{B}^{(v)}_w
\end{array}\right)
\end{equation*}
obtained by picking one column in the first block $\mat{B}^{(u)}_0$ and the other column in the second block $\mat{B}^{(v)}_w$.
As a consequence, $\dim I_2^{(u,v)}=(t-1)(t-q^w)$ if $w\leq e$, and $0$ otherwise.

Likewise if $u=v$, a basis of $I_2^{(u,u)}$ is given by the $2\times 2$ minors of $\mat{B}^{(u)}_0$,
and $\dim I_2^{(u,v)}=\frac{(t-1)(t-2)}{2}$.
\end{lemma}
\begin{proof}
First consider the case $v=u-w \mod m$ with $1\leq w\leq m/2$. Let
\beq
q^{(u,v)}_{a,b}=X^{(u)}_aX^{(v)}_{b+q^w}-X^{(v)}_bX^{(u)}_{a+1}
\eeq
be the $2\times 2$ minor constructed
from the column $\left[\begin{array}{c}X^{(u)}_a\\X^{(u)}_{a+1}\end{array}\right]$ of $\mat{B}^{(u)}_0$
and the column $\left[\begin{array}{c}X^{(v)}_b\\X^{(v)}_{b+q^w}\end{array}\right]$ of $\mat{B}^{(v)}_w$.
Then keeping the notation of \eqref{qabcd} we have $q^{(u,v)}_{a,b}=q^{(u,v)}_{a,b,c,d}$ with $c=a+1$ and $d=b+q^w$,
so condition \eqref{condabcd} is satisfied and $q^{(u,v)}_{a,b}\in I_2^{(u,v)}$.

Now we show that these $q^{(u,v)}_{a,b}$ generate $I_2^{(u,v)}$. For this it suffices to show that all the
$q^{(u,v)}_{a,b,c,d}$ (for $a,b,c,d$ satisfying \eqref{condabcd}) are linear combinations of the $q^{(u,v)}_{a,b}$.
After a possible change of sign we can assume $c>a$. Then we proceed by induction on the difference $c-a$.
First, in the case $c=a+1$, we necessarily have $d=b+q^w$, and $q^{(u,v)}_{a,b,c,d}=q^{(u,v)}_{a,b}$.
Now for general $c>a$ we write
\beq
q^{(u,v)}_{a,b,c,d}=q^{(u,v)}_{a,b+q^w,c-1,d}+q^{(u,v)}_{c-1,b}
\eeq
and apply the induction hypothesis to $q^{(u,v)}_{a,b+q^w,c-1,d}$.

Last, linear independence of the $q^{(u,v)}_{a,b}$ follows from the theory of $1$-genericity, but it can also
be proved directly as follows.
Consider a linear combination $\sum_{a,b}\lambda_{a,b}q^{(u,v)}_{a,b}$, and choose a minimal $a$ such that $\lambda_{a,b}\neq0$.
Then $q^{(u,v)}_{a,b}$ contains the monomial $X^{(u)}_aX^{(v)}_{b+q^w}$, but this monomial does not appear
in the other $q^{(u,v)}_{a',b'}$ with $\lambda_{a',b'}\neq0$ (indeed, either $a'=a$ but then $b'\neq b$, or $a'\geq a+1$ by minimality).
Thus the linear combination is nonzero.

The proof of the case $u=v$ is entirely similar, details are left to the reader.
\end{proof}

Now we translate these results in terms of the $2\times 2$ minors of the conjugate matrices $\mat{\Phi}^{(u)}$.
For $1\leq i_1<i_2\leq f$ we denote by $q^{(u)}_{i_1,i_2}$ the minor defined by the $i_1$-th and $i_2$-th columns
of $\mat{\Phi}^{(u)}=\left(\begin{array}{c|c|c|c}\mat{B}^{(u+e)}_0 & \mat{B}^{(u+e-1)}_1 & \cdots & \mat{B}^{(u)}_e\end{array}\right)$.
Observe that if $i_1$ falls in the block $\mat{B}^{(u+e-j_1)}_{j_1}$ of $\mat{\Phi}^{(u)}$,
and $i_2$ in the block $\mat{B}^{(u+e-j_2)}_{j_2}$, then
\beq
q^{(u)}_{i_1,i_2}\in I_2^{(u+e-j_1,u+e-j_2)}.
\eeq
It follows:
\begin{proposition}
For a given $u$, the minors $q^{(u)}_{i_1,i_2}$ with first column in $\mat{B}^{(u+e)}_0$
(i.e. with $j_1=0$, or equivalently $i_1\leq t-1$) are linearly independent.
Actually, they form a basis of the ``primitive'' subspace
\beq
I_2(\mat{\Phi}^{(u)})_{\operatorname{prim}}=I_2^{(u+e,u+e)}\oplus I_2^{(u+e,u+e-1)}\oplus\cdots\oplus I_2^{(u+e,u)}\subset I_2(\mat{\Phi}^{(u)}),
\eeq
of dimension $\dim I_2(\mat{\Phi}^{(u)})_{\operatorname{prim}}=\binom{f}{2}-\binom{f-(t-1)}{2}$.
\end{proposition}
\begin{proof}
Direct consequence of Lemmata~\ref{I2uvindep} and~\ref{baseI2uv}.
For the dimension, we count $\binom{f}{2}$ pairs $i_1,i_2$ with $1\leq i_1<i_2\leq f$,
and we exclude the $\binom{f-(t-1)}{2}$ of them with $i_1>t-1$.
\end{proof}
\begin{proposition}
With notation as above, the minors $q^{(u)}_{i_1,i_2}$ with first column not in $\mat{B}^{(u+e)}_0$
(i.e. with $j_1>0$, or equivalently, $i_1>t-1$) belong to
the intersection $I_2(\mat{\Phi}^{(u)})\cap I_2(\mat{\Phi}^{(u-j_1)})_{\operatorname{prim}}$.
\end{proposition}
\begin{proof}
Actually, $q^{(u)}_{i_1,i_2}\in I_2(\mat{\Phi}^{(u)})\cap I_2^{(u+e-j_1,u+e-j_2)}$.
\end{proof}
We can now answer Question (Q1) for $r=2$:
\begin{corollary}\label{corV2}
If $e<m/2$, the sum $V_2=I_2(\mat{\Phi}^{(0)})+\cdots+I_2(\mat{\Phi}^{(m-1)})$ can be written as a direct sum
\beq
V_2=I_2(\mat{\Phi}^{(0)})_{\operatorname{prim}}\oplus\cdots\oplus I_2(\mat{\Phi}^{(m-1)})_{\operatorname{prim}}
\eeq
so
\beq\label{dimV2}
\begin{split}
\beta_{1,2}(\code)\geq\dim V_2&=m\left(\binom{f}{2}-\binom{f-(t-1)}{2}\right)\\
&=\frac{m(t-1)}{2}\left((2e+1)t-2\frac{q^{e+1}-1}{q-1}\right).
\end{split}
\eeq
A basis of $V_2$ is given by the quadratic forms $q^{(u)}_{i_1,i_2}$ with $u,i_1,i_2$ ranging among those with $0\leq u\leq m-1$, $1\leq i_1<i_2\leq f$, and $i_1\leq t-1$.
\end{corollary}
\begin{proof}
Follows by the two Propositions above, and Lemma~\ref{I2uvindep} again.
The condition $e<m/2$ is there to ensure that no $I_2^{(u,v)}$ appears twice when we expand the sum.
\end{proof}
The lower bound \eqref{dimV2} was first observed without proof as~\cite[\emph{Exp.~fact~1}]{FGOPT13},
and was eventually proved as~\cite[Th.~20]{MorTil23}.
Actually, \cite{FGOPT13} makes computations with quadratic forms very similar to ours, but what lacks there
is a proof of linear independence (or said otherwise, in order to prove linear independence, the authors have to
assume that \eqref{dimV2} holds and is an equality).
On the other hand, the proof in~\cite{MorTil23} uses a different approach: the authors do not work with the space
of quadratic relations (the kernel of the degree~$2$ evaluation map), but with the square code
(the image of this evaluation map).
So what we just did is provide an alternative proof for~\cite[Th.~20]{MorTil23}, while reverting to the original approach
of~\cite{FGOPT13}.

As for Question (Q2), it is experimentally observed in \cite{FGOPT13} that \eqref{dimV2} is an equality
with overwhelming probability, so for $r=2$ the answer to (Q2) is positive in general.

Now we turn to the case $r\geq3$, where we can only propose conjectures.

Proposition~\ref{explicitEN} gives a basis $s_{r,j;i_1,\dots,i_r}(\mat{\Phi}^{(u)})$ for each syzygy space $M_{r-1,r}(\mat{\Phi}^{(u)})$,
where the indices $i_1,\dots,i_r$ correspond to columns of $\mat{\Phi}^{(u)}$.
A natural generalization of the answer to (Q1) would be as follows:
\begin{itemize}
\item For a given $u$, the $s_{r,j;i_1,\dots,i_r}(\mat{\Phi}^{(u)})$ with column $i_1$ in $\mat{B}^{(e+u)}_0$
are linearly independent, hence form a basis of a subspace $M_{r-1,r}(\mat{\Phi}^{(u)})_{\operatorname{prim}}\subset M_{r-1,r}(\mat{\Phi}^{(u)})$, of dimension $\binom{f}{r}-\binom{f-(t-1)}{r}$.
\item On the other hand, those with $i_1$ not in $\mat{B}^{(e+u)}_0$ belong to $M_{r-1,r}(\mat{\Phi}^{(u)})\cap\sum_{u'\neq u} M_{r-1,r}(\mat{\Phi}^{(u')})_{\operatorname{prim}}$.
\item If $e<m/2$, the sum $V_r=M_{r-1,r}(\mat{\Phi}^{(0)})+\cdots+M_{r-1,r}(\mat{\Phi}^{(m-1)})$ can
be written as the direct sum of the $M_{r-1,r}(\mat{\Phi}^{(u)})_{\operatorname{prim}}$.
Hence
\beq\label{improvedlowerboundbetaalt}
\beta_{r-1,r}(\code)\;\geq\;\dim(V_r)=m(r-1)\left(\binom{f}{r}-\binom{f-(t-1)}{r}\right)
\eeq
and in particular for $r>f-(t-1)$:
\beq\label{improvedlowerboundbetaalthigh}
\beta_{r-1,r}(\code)\;\geq\;\dim(V_r)=m(r-1)\binom{f}{r}.
\eeq
\end{itemize}
Although most of the arguments used to prove the case $r=2$ seem to generalize to $r\geq3$,
the linear algebra manipulations become more and more intricate, and the author was unfortunately
not able to derive a satisfactory proof.

On the other hand, we will see examples where \eqref{improvedlowerboundbetaalthigh}
not only holds but is an equality, which suggests that the
answer to (Q2) can be positive in some cases.
However the following example shows that for $r=3$ the answer to (Q2)
is negative in general, and the lower bound \eqref{improvedlowerboundbetaalt} still isn't tight:
\begin{example}
With $\code$ as above, $I_2(\code)$ contains the four quadratic forms
\beq X_{b+1}^{(u)}X_{c+1}^{(u+v)}-X_{b+q^v+1}^{(u)}X_{c}^{(u+v)} \eeq
\beq X_{b}^{(u)}X_{c+1}^{(u+v)}-X_{b+q^v}^{(u)}X_{c}^{(u+v)} \eeq
\beq X_{a}X_{b+q^v+1}^{(u)}-X_{a+q^u}X_{b+q^v}^{(u)} \eeq
\beq X_{a}X_{b+1}^{(u)}-X_{a+q^u}X_{b}^{(u)} \eeq
which admit the syzygy
\begin{equation}
\begin{split}
X_{a}(&X_{b+1}^{(u)}X_{c+1}^{(u+v)}-X_{b+q^v+1}^{(u)}X_{c}^{(u+v)})-X_{a+q^u}(X_{b}^{(u)}X_{c+1}^{(u+v)}-X_{b+q^v}^{(u)}X_{c}^{(u+v)})+\\
&+X_{c}^{(u+v)}(X_{a}X_{b+q^v+1}^{(u)}-X_{a+q^u}X_{b+q^v}^{(u)})-X_{c+1}^{(u+v)}(X_{a}X_{b+1}^{(u)}-X_{a+q^u}X_{b}^{(u)})=0.
\end{split}
\end{equation}
However this syzygy defines an element of $M_{2,3}(\code)$ that does not belong to $V_3$ in general.
\end{example}

For general $s$, the author only has guesses motivated by experimental evidence.
More precisely, in many examples, one finds that the lower bound \eqref{lowerboundrmaxalt} is an equality, i.e. 
\beq\label{tightrmaxalt}
\rmax(\code_s)=f-s
\eeq
and then also that the natural generalization of \eqref{improvedlowerboundbetaalthigh} holds and is an equality, i.e.
\beq\label{betaeqVr}
\beta_{r-1,r}(\code_s)=m(r-1)\binom{f-s}{r}
\eeq
for $r$ close enough to $f-s$.

Figures~\ref{Alt2105} and~\ref{Alt356} illustrate these facts.
Boldface values match~\eqref{betaeqVr} (and \eqref{tightrmaxalt}).

\begin{figure}[!h]
\vspace{-\baselineskip}
\begin{minipage}[b]{0.49\textwidth}
\begin{equation*}
\arraycolsep=2pt
\begin{array}{c|ccccccc}
s & \beta_{1,2} & \beta_{2,3} & \beta_{3,4} & \beta_{4,5} & \beta_{5,6} & \beta_{6,7} & \beta_{7,8}\\
\hline
0 & 251 & 1400 & 3230 & 2480 & \mathbf{1400} & \mathbf{480} & \mathbf{70}\\
1 & 202 & 880 & 1170 & \mathbf{840} & \mathbf{350} & \mathbf{60} & -\\
2 & 154 & 440 & \mathbf{450} & \mathbf{240} & \mathbf{50} & - & -\\
3 & 107 & \mathbf{200} & \mathbf{150} & \mathbf{40} & - & - & -\\
4 & 66 & \mathbf{80} & \mathbf{30} & - & - & - & -\\
5 & 31 & \mathbf{20} & - & - & - & - & -\\
6 & \mathbf{10} & - & - & - & - & - & -\\
7 & - & - & - & - & - & - & -
\end{array}
\end{equation*}
\vspace{-1.8\baselineskip}\caption{$s$-shortened $\Alt^\perp_{2,10,5}\;$ ($e\!=\!2,r^*\!=\!8$)}\label{Alt2105}
\end{minipage}
$\quad$
\begin{minipage}[b]{0.49\textwidth}
\begin{equation*}
\arraycolsep=2pt
\begin{array}{c|ccccccc}
s & \beta_{1,2} & \beta_{2,3} & \beta_{3,4} & \beta_{4,5} & \beta_{5,6} & \beta_{6,7} & \beta_{7,8}\\
\hline
0 & 222 & 1943 & 1725 & \mathbf{1120} & \mathbf{700} & \mathbf{240} & \mathbf{35}\\
1 & 193 & 1344 & \mathbf{525} & \mathbf{420} & \mathbf{175} & \mathbf{30} & -\\
2 & 165 & 801 & \mathbf{225} & \mathbf{120} & \mathbf{25} & - & -\\
3 & 138 & 312 & \mathbf{75} & \mathbf{20} & - & - & -\\
4 & 112 & \mathbf{40} & \mathbf{15} & - & - & - & -\\
5 & 87 & \mathbf{10} & - & - & - & - & -\\
6 & 63 & - & - & - & - & - & -\\
7 & 40 & - & - & - & - & - & -\\
\end{array}
\end{equation*}
\vspace{-1.8\baselineskip}\caption{$s$-shortened $\Alt^\perp_{3,5,6}\;$ ($e\!=\!1,r^*\!=\!8$)}\label{Alt356}
\end{minipage}
\end{figure}

This strongly suggests that the minimal resolution of $\code$
(resp. its shortened subcodes $\code_s$) contains $m$ conjugate Eagon-Northcott subcomplexes of length $f$ (resp. $f-s$),
and that in high enough degree this minimal resolution actually coincides with the direct sum of these conjugate subcomplexes.

\subsection*{Improved lower bounds in the Goppa case}

We just saw that the bound \eqref{lowerboundrmaxalt} on $\rmax$ is tight in general for alternant codes.
However one can improve it in the Goppa case, still using the same techniques. Let us focus on $q=2$ for simplicity.
Let $\phi:a\mapsto a^2$ be the Frobenius map, acting on any $\F_2$-algebra.
For any polynomial $g\in\F_{2^m}[X]$, set $L_g=g(X)^{-1}\F_{2^m}[X]_{<\deg(g)}\subset\F_{2^m}(X)$.
\begin{lemma}
Let $g(X)\in\F_{2^m}[X]$ be squarefree (i.e. separable). Then
\beq
L_g\,+\,\phi(L_g)\;=\;L_{g^2},
\eeq
the sum being direct.
\end{lemma}
\begin{proof}
Set $t=\deg(g)$. Then $L_g$ and $\phi(L_g)$ both are $t$-dimensional subspaces of the $2t$-dimensional space $L_{g^2}$.
To conclude we only have to prove $L_g\cap\phi(L_g)=0$.
However we have $F(X)\in L_g\cap\phi(L_g)$ if and only if $F(X)=\frac{A(X)}{g(X)}=\frac{B(X)^2}{g(X)^2}$
for some $A,B$ of degree $<t$. But then this implies $g(X)|B(X)^2$ with $g$ squarefree of degree $t$,
which is impossible unless $B=0$.
\end{proof}
From this we readily deduce:
\begin{lemma}
Let $\mot{x}\in(\F_{2^m})^n$ be a support, and $g(X)\in\F_{2^m}[X]$ squarefree of degree $t\leq n/2$,
not vanishing on any entry of $\mot{x}$. Then
\beq
\GRS_t(\mot{x},g(\mot{x})^{-1})+\GRS_t(\mot{x}^2,g(\mot{x})^{-2})=\GRS_{2t}(\mot{x},g(\mot{x})^{-2}),
\eeq
the sum being direct.
\end{lemma}
\begin{proposition}\label{decompose_Goppa_binaire}
Let $\code=\Gop(\mot{x},g)^\perp\in\Gop^{\operatorname{sqfr},\perp}_{2,m,n,t}$ be proper, with $g$ squarefree.
Set $\mot{y}=g(\mot{x})^{-1}$.
Then
\beq
\code_{\F_{q^m}}=\bigoplus_{i=0}^{m/2-1}\GRS_{2t}(\mot{x}^{4^i},\mot{y}^{4^i}),\quad\text{or}
\eeq
\beq
\code_{\F_{q^m}}=\left(\bigoplus_{i=0}^{(m-1)/2-1}\GRS_{2t}(\mot{x}^{4^i},\mot{y}^{4^i})\right)\oplus\GRS_{t}(\mot{x}^{2^{m-1}},\mot{y}^{2^{m-1}})
\eeq
depending on whether $m$ is even or odd.
\end{proposition}
\begin{corollary}\label{corlowerboundGop}
Let $\code=\Gop(\mot{x},g)^\perp\in\Gop^{\operatorname{sqfr},\perp}_{2,m,n,t}$ be proper, with $g$ squarefree.
Set $\widehat{e}=\lfloor\log_4(2t-1)\rfloor$
and $\widehat{f}=(2\widehat{e}+2)t-\frac{4^{\widehat{e}+1}-1}{3}$.
For any $s\geq 0$, let $\code_s$ be a $s$-shortened subcode of $\code$.
Then for all $r\geq 2$ we have $\beta_{r-1,r}(\code_s)\geq(r-1)\binom{\widehat{f}-s}{r}$,
hence
\beq\label{lowerboundrmaxGop}
\rmax(\code_s)\geq \widehat{f}-s.
\eeq
\end{corollary}
\begin{proof}
Same as Theorem~\ref{thlowerboundalt}, with $\mat{\Phi}$ adapted to fit Proposition~\ref{decompose_Goppa_binaire}.
\end{proof}

\begin{remark}\label{improvedcorlowerboundGop}
It turns out that this last lower bound can be improved still further, replacing the value of $\widehat{f}$ with
\beq
r^*=(a+2)t-2^a
\eeq
where $a=\lceil\log_2(t)\rceil$.
The proof, which will be found in the forthcoming paper~\cite{stability}, uses new ingredients,
the main of which being a reinterpretation of the number of columns of a convenient matrix $\mat{\Phi}$
as $r^*=\dim(\code_{\F_{2^m}}\cap\mot{x}\code_{\F_{2^m}})$.

(A similar result also holds in the nonbinary case, but then we have to set $r^*=(a+1)t-q^a$
where $a=\lceil\log_q(t/(q-1))\rceil$. Observe that the binary case does specialize from
the general case with $q=2$, as the $a+1$ becomes a $a+2$.)
\end{remark}

This new bound now appears to be tight in most examples, i.e.
\beq\label{tightrmaxGop}
\rmax(\code_s)=r^*-s
\eeq
with $r^*$ given by Remark~\ref{improvedcorlowerboundGop}.
Moreover, the analogue of \eqref{betaeqVr} then also appears to hold experimentally, i.e.
\beq\label{betaGop}
\beta_{r-1,r}(\code_s)=m(r-1)\binom{r^*-s}{r}
\eeq
for $r$ close enough to $r^*-s$.

Figures~\ref{Gop285}-\ref{Gop21010} illustrate these facts.
Boldface values match~\eqref{betaGop} (and~\eqref{tightrmaxGop}),
leading to the same presumption that the one following Figures~\ref{Alt2105} and~\ref{Alt356}.

\begin{figure}[!h]
\begin{equation*}
\arraycolsep=2pt
\begin{array}{c|ccccccccc}
s & \beta_{1,2} & \beta_{2,3} & \beta_{3,4} & \beta_{4,5} & \beta_{5,6} & \beta_{6,7} & \beta_{7,8} & \beta_{8,9}\\
\hline
\vdots & \vdots & \vdots & \vdots & \vdots & \vdots & \vdots & \vdots & \vdots\\
8 & 280 & 3224 & 7464 & 4272 & \mathbf{3360} & \mathbf{1728} & \mathbf{504} & \mathbf{64}\\
9 & 249 & 2510 & 1800 & \mathbf{1792} & \mathbf{1120} & \mathbf{384} & \mathbf{56} & -\\
10 & 219 & 1856 & \mathbf{840} & \mathbf{672} & \mathbf{280} & \mathbf{48} & - & -\\
11 & 190 & 1260 & \mathbf{360} & \mathbf{192} & \mathbf{40} & - & - & -\\
12 & 162 & 720 & \mathbf{120} & \mathbf{32} & - & - & - & -\\
13 & 135 & 234 & \mathbf{24} & - & - & - & - & -\\
14 & 109 & \mathbf{16} & - & - & - & - & - & -
\end{array}
\end{equation*}
\vspace{-1.5\baselineskip}\caption{$s$-shortened $\Gop^\perp_{2,8,5}$ ($a\!=\!3,r^*\!=\!17$)}\label{Gop285}
\vspace{\baselineskip}
\begin{minipage}{0.52\textwidth}
\begin{equation*}
\arraycolsep=2pt
\begin{array}{c|cccccc}
s & \beta_{1,2} & \beta_{2,3} & \beta_{3,4} & \beta_{4,5} & \beta_{5,6} & \beta_{6,7}\\
\hline
\vdots & \vdots & \vdots & \vdots & \vdots & \vdots & \vdots\\
32 & 719 & 8474 & \mathbf{450} & \mathbf{240} & \mathbf{50} & -\\
33 & 662 & 6216 & \mathbf{150} & \mathbf{40} & - & -\\
34 & 606 & 4070 & \mathbf{30} & - & - & -\\
35 & 551 & 2034 & - & - & - & -\\
\end{array}
\end{equation*}
\vspace{-1.5\baselineskip}\caption{$s$-shortened $\Gop^\perp_{2,10,9}$ ($a\!=\!4,r^*\!=\!38$)}\label{Gop2109}
\end{minipage}
\begin{minipage}{0.52\textwidth}
\begin{equation*}
\arraycolsep=2pt
\begin{array}{c|ccccc}
s & \beta_{1,2} & \beta_{2,3} & \beta_{3,4} & \beta_{4,5} & \beta_{5,6}\\
\hline
\vdots & \vdots & \vdots & \vdots & \vdots & \vdots\\
40 & 846 & 13924 & \mathbf{30} & - & - \\
41 & 787 & 11426 & - & - & - \\
\end{array}
\end{equation*}
\vspace{-1.2\baselineskip}\caption{$s$-shortened $\Gop^\perp_{2,10,10}$ ($a\!=\!4,r^*\!=\!44$)}\label{Gop21010}
\end{minipage}
\vspace{-\baselineskip}\end{figure}

\noindent Observe that Theorem~\ref{thlowerboundalt} and Corollary~\ref{corlowerboundGop} are not tight on these examples:
\begin{itemize}
\item for $\Gop^\perp_{2,8,5}$ they give $f=8$ and $\widetilde{f}=15$, compared to $r^*=17$
\item for $\Gop^\perp_{2,10,9}$ they give $f=21$ and $\widetilde{f}=33$, compared to $r^*=38$
\item for $\Gop^\perp_{2,10,10}$ they give $f=25$ and $\widetilde{f}=39$, compared to $r^*=44$
\end{itemize}
so Remark~\ref{improvedcorlowerboundGop} really is an improvement.

\section{Regularity $2$ and the small defect heuristic}\label{secdefect}

\subsection*{Regularity $2$ and consequences}

If $\code$ is a $[n,k]$-code, we let
\beq\label{defB}
B_j=\sum_{i\geq0}(-1)^i\beta_{i,j}
\eeq
be the alternating sum of its Betti numbers degree $j$,
and $B(z)=\sum_{j\geq0}B_jz^j$ their generating polynomial (it is indeed a finite sum).

Let also $H_{\code}(z)=\sum_{r\geq0}z^r\dim\code\deux[r]$ be the Hilbert series of $\code$.
\begin{proposition}\label{BH}
We have
\beq
B(z)=(1-z)^kH_{\code}(z).
\eeq
\end{proposition}
\begin{proof}
Generating series reformulation of \cite[Cor.~1.10]{Eisenbud05}.
\end{proof}

\begin{definition}[{\cite[Def.~1.5 \& Th.~2.35]{HR15}}]
The Castelnuovo-Mumford regularity of a projective $[n,k]$-code $\code$ is the smallest integer $r$
such that $\code\deux[r]=\F^n$.
\end{definition}
\begin{definition}[{cf. \cite[\S4A]{Eisenbud05}, after \cite[Lect.~14]{Mumford66}}]\label{regCM}
The Castelnuovo-Mumford regularity of $\code\deux[\cdot]$ is $\max\{r:\:\exists i,\,\beta_{i,i+r}(\code)>0\}$.
\end{definition}
\begin{proposition}[{\cite[Th.~4.2]{Eisenbud05}}]
These two definitions coincide.
\end{proposition}

From now on we will just say ``regularity'' for short.

The square code distinguisher, and the filtration attack from \cite{BarMorTil24} that extends it,
work for codes $\code$ with $\code\deux\subsetneq\F^n$, i.e. of regularity $>2$.
This means that codes of regularity $2$ are hard to deal with under this approach.
On the opposite, for us, codes of regularity $2$ are nice because Definition~\ref{regCM}
means their Betti diagram is simple: it has only two nontrivial rows.
Observe that most codes of interest have regularity $2$ (with the notable exception of self-dual codes).

\begin{definition}\label{definddef}
If $f:U\longto V$ is a linear map between finite dimensional $\F$-vector spaces, we define its index
\beq
\begin{split}
\ind(f)&=\dim(U)-\dim(V)\\
&=\dim\ker(f)-\dim\coker(f)
\end{split}
\eeq
and its rank defect
\beq
\begin{split}
\defect(f)&=\min(\dim(U),\dim(V))-\rk(f)\\
&=\min(\dim\ker(f),\dim\coker(f)).
\end{split}
\eeq
\end{definition}
For any real $x$ we set $x^+=\max(x,0)$ and $x^-=(-x)^+$, so $x=x^+-x^-$.
Then we always have
\begin{align}
\label{ker>ind+}\dim\ker(f)&\geq\ind(f)^+\\
\label{coker>ind-}\dim\coker(f)&\geq\ind(f)^-
\end{align} 
and then
\beq
\defect(f)=\dim\ker(f)-\ind(f)^+=\dim\coker(f)-\ind(f)^-
\eeq
measures the distance to equality in these two inequalities.

For $r\geq3$, recall the linear map $\phi_r: M_{r-2,r-1}\tens S_1 \longto M_{r-2,r}$ from \eqref{defphir}.
It will also be handy to let $\phi_2=\ev_{\mat{G},2}:S_2\longto\code\deux$ be the evaluation map.
\begin{theorem}\label{indphi}
For all $r\geq2$ we have
\begin{align}
\dim\ker(\phi_r)&=\beta_{r-1,r}\label{betaker}\\
\dim\coker(\phi_r)&=\beta_{r-2,r}\label{betacoker}
\end{align} 
and moreover, if $\code\deux=\F^n$, then
\beq\label{formuleindphi}
\ind(\phi_r)=\left(\frac{k(k+1)}{r}-n\right)\binom{k-1}{r-2}
\eeq
\end{theorem}
\begin{proof}
For $r=2$ this is proved directly. Now assume $r\geq3$.
Then the first two equalities are reformulations of Lemma~\ref{kercokerphir}.
Now assume moreover $\code\deux=\F^n$.
Then $H_{\code}(z)=1+kz+n\frac{z^2}{1-z}$,
so $B(z)=(1+kz)(1-z)^k+nz^2(1-z)^{k-1}$ by Proposition~\ref{BH},
thus $B_r=(-1)^{r-1}(r-1)\binom{k+1}{r}+(-1)^rn\binom{k-1}{r-2}=(-1)^{r-1}\left(\frac{k(k+1)}{r}-n\right)\binom{k-1}{r-2}$.
On the other hand, $\code\deux=\F^n$ also means $\code$ has regularity $2$ in the sense of Definition~\ref{regCM},
so \eqref{defB} reduces to $B_r=(-1)^{r-1}\beta_{r-1,r}+(-1)^{r-2}\beta_{r-2,r}=(-1)^{r-1}\ind(\phi_r)$
and we conclude.
\end{proof}
\begin{corollary}
If $\code\deux=\F^n$, then the bottom right entry of its Betti diagram is
\beq
\beta_{k-1,k+1}=n-k.
\eeq
\end{corollary}
\begin{proof}
The minimal resolution of $\code$ has length $k-1$, so $\beta_{i,j}=0$ for $i>k-1$,
hence $\beta_{k-1,k+1}=-\ind(\phi_{k+1})=n-k$.
\end{proof}

Regularity $2$ helps in computing the full Betti diagram of codes, as one row can be deduced from the other.
In order to illustrate this, let us first recall that a $[n,k]$-code $\code$ is MDS
if it has dual minimum distance $\dmin(\code^\perp)=k+1$.
\begin{lemma}[{\cite[Th.~1]{GreLaz88}, reformulated}]\label{GLre}
Let $\code$ be a $[n,k]$-code with $n\leq 2k-1$. Assume $\code$ is MDS.
Then for all $r\leq 2k+1-n$ we have $\beta_{r-2,r}(\code)=0$.
\end{lemma}
\begin{proposition}\label{Bettiparite}
Let $\code$ be a $[k+1,k]$ MDS code, for instance a parity code or a $[k+1,k]$ $\GRS$ code.
Then the nonzero Betti numbers of $\code$ are $\beta_{0,0}=\beta_{k-1,k+1}=1$, and
\beq
\beta_{r-1,r}=\frac{(r-1)(k-r)}{k}\binom{k+1}{r}
\eeq
for $2\leq r\leq k-1$. In particular they satisfy the symmetry $\beta_{i,j}=\beta_{k-1-i,k+1-j}$.
\end{proposition}
\begin{proof}
The parameters imply that $\code$ has regularity $2$.
By Lemma~\ref{GLre} we have $\beta_{r-2,r}=0$ for $r\leq k$.
Then $\beta_{r-1,r}=\ind(\phi_r)=\left(\frac{k(k+1)}{r}-(k+1)\right)\binom{k-1}{r-2}$
and we conclude with a straightforward calculation.
\end{proof}
\begin{proposition}\label{BettiGRScritique}
Let $\code$ be a $[2k-1,k]$ $\GRS$ code.
Then the nonzero Betti numbers of $\code$ are:
\begin{itemize}
\item $\beta_{0,0}=1$
\item $\beta_{r-1,r}=(r-1)\binom{k-1}{r}\;$ for $2\leq r\leq k-1$
\item $\beta_{r-2,r}=(r-2)\binom{k-1}{r-2}\;$ for $3\leq r\leq k+1\quad$ (so $\beta_{r-2,r}=\beta_{k+2-r,k+4-r}$).
\end{itemize}
In particular the ideal $I(\code)$ is generated by $\beta_{1,2}=\frac{(k-1)(k-2)}{2}$ quadratic forms
and $\beta_{1,3}=k-1$ cubic forms.
\end{proposition}
\begin{proof}
Using the basis $\mot{y},\mot{y}\mot{x},\dots,\mot{y}\mot{x}^{k-1}$ of $\code$,
we see $I_2(\code)$ contains the $2\times 2$ minors
of $\mat{\Phi}=\left(\begin{array}{cccc}X_0 & X_1 & \dots & X_{k-2}\\ X_1 & X_2 & \dots & X_{k-1}\end{array}\right)$,
i.e. $I_2(\code)\supset\im(\bigwedge^2\mat{\Phi})$ of dimension $\binom{k-1}{2}$.

Now set $n=2k-1$.

Because $n\geq 2k-1$, we have $\dim\code\deux=2k-1$, hence $\dim I_2(\code)=\binom{k-1}{2}$, so $I_2(\code)=\im(\bigwedge^2\mat{\Phi})$.
Corollary~\ref{I2_determine_tout} and the Eagon-Northcott complex of $\mat{\Phi}$ then give $\beta_{r-1,r}=(r-1)\binom{k-1}{r}$.

Because $n\leq 2k-1$, we have $\code\deux=\F^n$.
Then $\beta_{r-2,r}=\beta_{r-1,r}-\ind(\phi_r)=(r-1)\binom{k-1}{r}-\left(\frac{k(k+1)}{r}-(2k-1)\right)\binom{k-1}{r-2}$,
and a straightforward calculation allows to conclude.
\end{proof}
Figures~\ref{figureparite} and~\ref{figureGRScritique} illustrate Propositions~\ref{Bettiparite} and~\ref{BettiGRScritique}.
\begin{figure}[h]
\vspace{-2.5\baselineskip}
\begin{minipage}[b]{0.49\textwidth}
\begin{equation*}
\arraycolsep=3pt
\begin{array}{c|cccccccc}
& 0 & 1 & 2 & 3 & 4 & 5 & 6 & 7 \\
\hline
0 & 1 & - & - & - & - & - & - & - \\
1 & - & 27 & 105 & 189 & 189 & 105 & 27 & - \\
2 & - & - & - & - & - & - & - & 1
\end{array}
\end{equation*}
\vspace{-2\baselineskip}\caption{a $[9,8]$ parity or $\GRS$ code}\label{figureparite}
\end{minipage}
\begin{minipage}[b]{0.49\textwidth}
\begin{equation*}
\arraycolsep=3pt
\begin{array}{c|ccccccccc}
& 0 & 1 & 2 & 3 & 4 & 5 & 6 & 7 \\
\hline
0 & 1 & - & - & - & - & - & - & - \\
1 & - & 21 & 70 & 105 & 84 & 35 & 6 & - \\
2 & - & 7 & 42 & 105 & 140 & 105 & 42 & 7
\end{array}
\end{equation*}
\vspace{-2\baselineskip}\caption{a $[15,8]$ $\GRS$ code}\label{figureGRScritique}
\end{minipage}
\vspace{-1.5\baselineskip}
\end{figure}

\begin{exercise}
Compute the Betti diagram of $[n,k]$ $\GRS$ codes for $k+1<n<2k-1$.
\end{exercise}

\subsection*{The small defect heuristic}

Because of \eqref{betaker}, a necessary condition for a $[n,k]$-code $\code$ of regularity $2$
to have $\beta_{r-1,r}=0$ is that $\ind(\phi_r)\leq0$.
This immediately yields
\beq
\rmax(\code)\geq\lceil\rcrit(\code)\rceil-1
\eeq
where we define the \emph{critical degree} of $\code$
\beq
\rcrit(\code)=\frac{k(k+1)}{n}
\eeq
as the value of $r$ at which the formula \eqref{formuleindphi} for $\ind(\phi_r)$ changes sign.

More precisely, \eqref{ker>ind+} and Theorem~\ref{indphi} give the lower bound
\beq\label{contingent}
\beta_{r-1,r}(\code)\geq\left(\frac{k(k+1)}{r}-n\right)^+\binom{k-1}{r-2}
\eeq
that can be interpreted as the number of ``contingent'' syzygies that are forced to exist,
solely because of the dimension difference between the source and target spaces of the map $\phi_r$.
Then, when $\code$ is random, it is quite natural to expect that no other
syzygy will appear, or equivalently that \eqref{contingent} will be an equality, i.e. $\defect(\phi_r)=0$.

This is compatible with results such as those in \cite{Cooper00} and the references within,
that give estimates on the distribution of the rank, hence the defect, of random linear maps,
under various probability laws.
Unfortunately, even if the probability distribution of $\code$ is nice
(say, uniform among codes of given $[n,k]$), it is not easy to control the distribution of $\phi_r$,
so we will not be able to give proofs.
Moreover it could be that the iterative algebraic process in the construction of $\phi_r$
would lead to some unexpected constraints on its rank.
And indeed, in this section we will identify some parameter ranges for which the defect \emph{cannot} be small;
but conversely, we will also give arguments that support the validity of this small defect heuristic
\emph{at least under some proper conditions}.

A first argument in support of the heuristic is that it is unconditionally true when $r=2$.
Indeed, in this special case, \cite{CCMZ15} manages to give lower bounds, exponentially close to $1$,
on the probability that $\defect(\phi_2)=0$.

Another argument is the \emph{minimal resolution conjecture} of \cite{Lorenzini93}.
It postulates that, over an infinite field, a Zariski-generic code satisfies
$\beta_{r-1,r}(\code)=0$ for $r\geq\rcrit(\code)$
and $\beta_{r-2,r}(\code)=0$ for $r\leq\rcrit(\code)$,
or equivalently, $\defect(\phi_r)=0$ for all $r\geq2$.
However, two points require our attention:
\begin{enumerate}
\item\label{MRCfalse} This conjecture is now known to be false in general \cite{EisPop99}.
\item\label{finitevsinfinite} We work over a finite field, not an infinite one.
\end{enumerate}

Concerning point \ref{MRCfalse}, we will argue that the conjecture is still ``true enough'' for our use.
First, a nonzero defect might not be a problem in our Betti number estimates, as long as it remains small.
Moreover, as noted in the introduction of \cite{EisPop99}, the conjecture has been proved for a large range
of values of $n$ and $k$. In fact, although \cite{EisPop99} provides an infinity of counterexamples,
these remain limited to very specific parameters, namely of the form $n=k+O(\sqrt{k})$.
And indeed, perhaps the most valuable result for us is \cite{HirSim96}, which proves that the conjecture is true
when $n$ is large enough with respect to $k$.

Concerning point \ref{finitevsinfinite}, clearly, codes behave quite differently over a finite field
and over an infinite field.
For instance, generic codes over an infinite field are MDS for any parameter set, while over a finite field
they clearly are not.
This makes it desirable to investigate links between Betti numbers and distance properties of a code.
At this moment the author mostly has only experimental results in this direction:
Actually, experiments suggest much stronger results hold.
For any code $\code$, let us denote by $A_i(\code)$ the number of codewords of Hamming weight $i$ in $\code$.
Also recall $\rcrit=\frac{k(k+1)}{n}$.

\begin{expfact}\label{Betti_dmin_dduale}
Let $\code$ be a $[n,k]_q$-code of regularity $2$, with minimum distance $d=\dmin(\code)$ and dual minimum distance $d^\perp=\dmin(\code^\perp)$.
\begin{enumerate}
\item For all $r\geq d^\perp$ we have $\beta_{r-2,r}(\code)>0$.

Moreover if $d^\perp\leq\rcrit$ we have $\beta_{d^\perp-2,d^\perp}(\code)\geq A_{d^\perp}(\code^\perp)/(q-1)$ (and there are numerous examples of $\code$ where this inequality is an equality).
\item Dually, for all $r\leq k+1-d$ we have $\beta_{r-1,r}(\code)>0$.

Moreover if $k+1-d\geq\rcrit$ we have $\beta_{k-d,k+1-d}(\code)\geq A_{d}(\code)/(q-1)$ (and there are numerous examples of $\code$ where this inequality is an equality).
\end{enumerate}
It follows that $\defect(\phi_r)>0$ for $d^\perp\leq r\leq\rcrit$ and for $\rcrit\leq\ r\leq k+1-d$, when applicable (i.e. when these intervals are nonempty).
\end{expfact}
\begin{expfact}\label{converse}
Conversely, for a \emph{random} code $\code$ among codes of given parameters $[n,k,d,d^\perp]_q$,
the probability that $\defect(\phi_r)=0$ tends quickly to $1$
as $r\ll\min(d^\perp,\rcrit)$ and as $r\gg\max(\rcrit,k+1-d)$.

In particular $\rmax(\code)=\max(\lfloor\rcrit\rfloor,k+1-d)$, or is very close to this value, with high probability.
\end{expfact}

Actually, as a modest first step, the author is able to prove the first assertion in point~2. of Experimental fact~\ref{Betti_dmin_dduale}:
\begin{proposition}\label{rmax>k+1-d}
Any $[n,k,d]$ code $\code$ has $\rmax(\code)\geq k+1-d$,
i.e. $\beta_{r-1,r}(\code)>0$ for all $r\leq k+1-d$.
\end{proposition}
\begin{proof}
Let $\mot{c}_1$ be a codeword of minimum weight $\w(\mot{c}_1)=d$.
Let $S$ be the support of $\mot{c}_1$, and let $\code_S\subset\code$ be the subcode made
of all codewords that vanish over $S$. Then $\dim(\code_S)\geq k-\card{S}=k-d$, so we can find $k-d$ linearly
independent codewords $\mot{c}_2,\dots,\mot{c}_{k+1-d}$ in $\code_S$.
Let then complete $\mot{c}_1,\dots,\mot{c}_{k-d+1}$ into a basis $\mot{c}_1,\dots,\mot{c}_k$ of $\code$,
and use this basis to define the graded evaluation map $S=\F[X_1,\dots,X_k]\longto\code\deux[\cdot]$.
By construction we have $\mot{c}_1\mot{c}_2=\cdots=\mot{c}_1\mot{c}_{k-d+1}=\mot{0}$,
which means $\linspan{X_1X_2,\dots,X_1X_{k+1-d}}\subset I_2(\code)$.
By Corollary~\ref{sous-resolutions} we deduce $\beta_{r-1,r}(\code)\geq\beta_{r-1,r}(S/\linspan{X_1X_2,\dots,X_1X_{k+1-d}})$ for all $r$,
hence in particular $\rmax(\code)\geq\rmax(S/\linspan{X_1X_2,\dots,X_1X_{k+1-d}})$.

However we have $\linspan{X_1X_2,\dots,X_1X_{k+1-d}}=X_1\linspan{X_2,\dots,X_{k+1-d}}$, so the
minimal resolution of $\linspan{X_1X_2,\dots,X_1X_{k+1-d}}$ can be identified with the Koszul complex
on $X_2,\dots,X_{k+1-d}$ (with degree shifted by $1$), so its $\rmax$ is $k+1-d$.
\end{proof}

\begin{remark}
The quantity $\rmax=\max\{r:\;\beta_{r-1,r}>0\}$, and the
dual quantity $\rmin=\min\{r:\;\beta_{r-2,r}>0\}$ that implicitely appears together with it
in these Experimental facts, are at the core of deep conjectures of Green and Green-Lazarsfeld,
discussed e.g. in the last two chapters of \cite{Eisenbud05}
(with the notation therein, $\rmax=b(X)$ and $\rmin=a(X)+3$).
\end{remark}

Let us now illustrate various aspects of these Experimental facts.

Figures~\ref{figureMRC}-\ref{figurebinaryGolaybis} compare Betti diagrams
for three codes of length $n=23$ and dimension $k=12$
(so $\rcrit\approx 6.78$)
of quite different nature.

First, Figure~\ref{figureMRC} presents the Betti diagram as predicted by the minimal resolution conjecture,
i.e. with all defects equal to zero.
It turns out that the minimal resolution conjecture is true for these parameters,
so generic $[23,12]$-codes over an infinite field actually have this Betti diagram.
Such codes are MDS, i.e. they have $d=12$ and $d^\perp=13$,
so the intervals $d^\perp\leq r\leq\rcrit$ and $\rcrit\leq\ r\leq k+1-d$
are empty, and Experimental fact~\ref{Betti_dmin_dduale} is trivially verified.
Actually, codes with this Betti diagram can be found even over not-so-large finite fields, e.g. $q=5$ suffices (the author did not search for optimality).
These codes are no longer MDS, but they still have $d,d^\perp\geq7$ and again Experimental fact~\ref{Betti_dmin_dduale} is trivially verified.
Moreover, the converse Experimental fact~\ref{converse} is also verified as all defects are zero.

\begin{figure}[h]
\begin{equation*}
\arraycolsep=5pt
\begin{array}{c|cccccccccccc}
& 0 & 1 & 2 & 3 & 4 & 5 & 6 & 7 & 8 & 9 & 10 & 11\\
\hline
0 & 1 & - & - & - & - & - & - & - & - & - & - & -\\
1 & - & 55 & 319 & 880 & 1353 & 990 & - & - & - & - & - & -\\
2 & - & - & - & - & - & 330 & 1617 & 1870 & 1221 & 485 & 110 & 11\\
\end{array}
\end{equation*}
\vspace{-1.2\baselineskip}\caption{an idealized $[23,12]$-code according to the minimal resolution conjecture}\label{figureMRC}
\begin{equation*}
\arraycolsep=5pt
\begin{array}{c|cccccccccccc}
& 0 & 1 & 2 & 3 & 4 & 5 & 6 & 7 & 8 & 9 & 10 & 11\\
\hline
0 & 1 & - & - & - & - & - & - & - & - & - & - & -\\
1 & - & 55 & 319 & 884 & 1397 & 1224 & 490 & 121 & 18 & 1 & - & -\\
2 & - & - & 4 & 44 & 234 & 820 & 1738 & 1888 & 1222 & 485 & 110 & 11\\
\end{array}
\end{equation*}
\vspace{-1.2\baselineskip}\caption{a (pseudo)random $[23,12]_2$-code ($d=3$, $d^\perp=4$)}\label{figurepseudorandom}
\begin{equation*}
\arraycolsep=5pt
\begin{array}{c|cccccccccccc}
& 0 & 1 & 2 & 3 & 4 & 5 & 6 & 7 & 8 & 9 & 10 & 11\\
\hline
0 & 1 & - & - & - & - & - & - & - & - & - & - & -\\
1 & - & 55 & 320 & 891 & 1408 & 1210 & 320 & 55 & - & - & - & -\\
2 & - & 1 & 11 & 55 & 220 & 650 & 1672 & 1870 & 1221 & 485 & 110 & 11\\
\end{array}
\end{equation*}
\vspace{-\baselineskip}\caption{the $[23,12]_2$ Golay code ($d=7$, $d^\perp=8$)}\label{figurebinaryGolaybis}
\vspace{-1.5\baselineskip}
\end{figure}

Second, Figure~\ref{figurepseudorandom} presents the Betti diagram of the binary code whose generator matrix
\begin{equation*}
G=\left(\begin{array}{ccccccccccccccccccccccc}
1 & 1 & 0 & 0 & 1 & 0 & 0 & 1 & 0 & 0 & 0 & 0 & 1 & 1 & 1 & 1 & 1 & 1 & 0 & 1 & 1 & 0 & 1\\
0 & 1 & 0 & 1 & 0 & 0 & 0 & 1 & 0 & 0 & 0 & 1 & 0 & 0 & 0 & 0 & 1 & 0 & 1 & 1 & 0 & 1 & 0\\
 & & & & & \vdots & & & & & & & & & & & & \vdots & & & & & \\
0 & 1 & 0 & 0 & 1 & 0 & 1 & 0 & 0 & 0 & 1 & 0 & 1 & 0 & 0 & 1 & 0 & 1 & 0 & 0 & 0 & 0 & 0
\end{array}\right)
\end{equation*}
of size $12\times23$, is constructed from the binary expansion of $\pi$.
We consider this code as representative of random codes.
One can check that it has $d=3$, $d^\perp=4$.
Then Experimental fact~\ref{Betti_dmin_dduale} is verified, as the defect is positive for $4\leq r\leq 10$.
Moreover at the extremities of this interval we have equality $\beta_{2,4}=A_4(C^\perp)=4$ and $\beta_{9,10}=A_3(C)=1$, as predicted.
And the defect is zero for $r<4$ and $r>10$ so Experimental fact~\ref{converse} is verified also.

Last, Figure~\ref{figurebinaryGolaybis} reproduces Figure~\ref{figurebinaryGolay}, the binary Golay code.
As $d,d^\perp\geq7$, Experimental fact~\ref{Betti_dmin_dduale} is trivially verified.
On the other hand, Experimental fact~\ref{converse} is not verified, but this only illustrates the fact that the Golay code
is certainly \emph{not} representative of random codes. Indeed it has a lot of algebraic structure, which explains
it admitting special syzygies.

\vspace{\baselineskip}

More systematically, Figure~\ref{statdef} presents statistics on $\defect(\phi_r)$ ($2\leq r\leq 8$) for random $[56,16]_2$-codes.
For each pair $(d,d^\perp)$, a few thousands of codes with these parameters were sampled uniformly (using rejection sampling).
The average value of $\defect(\phi_r)$ among these samples is displayed, and also its $99\%$ distribution interval (which means at most $0.5\%$ fall below and at most $0.5\%$ above).

Here we have $\rcrit\approx 4.86$.
Then, in accordance with Experimental fact~\ref{Betti_dmin_dduale}, we can check $\defect(\phi_r)>0$ for $d^\perp\leq r\leq 4$ and for $5\leq r\leq 17-d$ (when applicable);
while $\defect(\phi_r)=0$ with higher and higher probability as we move away from these intervals,
in accordance with Experimental fact~\ref{converse}.

Experiments with other parameters display similar results.

\begin{figure}\vspace{-\baselineskip}
\begin{minipage}[b]{0.09\textwidth}
\hfill
\end{minipage}
\begin{minipage}[b]{0.22\textwidth}
\begin{equation*}
d=11
\end{equation*}
\end{minipage}
\begin{minipage}[b]{0.22\textwidth}
\begin{equation*}
d=12
\end{equation*}
\end{minipage}
\begin{minipage}[b]{0.22\textwidth}
\begin{equation*}
d=13
\end{equation*}
\end{minipage}
\begin{minipage}[b]{0.22\textwidth}
\begin{equation*}
d=14
\end{equation*}
\end{minipage}

\begin{minipage}[b][.18\textheight][c]{0.09\textwidth}
\begin{equation*}
d^\perp\!=\!3
\end{equation*}
\end{minipage}
\begin{minipage}[b][.18\textheight][c]{0.22\textwidth}
\begin{equation*}
\begin{array}{c|c|c}
r\, & \;\,\textrm{mean}\,\; & \;\;99\%\;\;\\
\hline
2 & 0.000 & [0,0]\\
3 & 1.269 & [1,3]\\
4 & 23.821 & [15,55]\\
5 & 6.927 & [5,21]\\
6 & 1.341 & [1,7]\\
7 & 0.042 & [0,1]\\
8 & 0.000 & [0,0]
\end{array}
\end{equation*}
\end{minipage}
\begin{minipage}[b][.18\textheight][c]{0.22\textwidth}
\begin{equation*}
\begin{array}{c|c|c}
r\, & \;\,\textrm{mean}\,\; & \;\;99\%\;\;\\
\hline
2 & 0.000 & [0,0]\\
3 & 1.245 & [1,3]\\
4 & 23.171 & [15,52]\\
5 & 1.948 & [1,8]\\
6 & 0.086 & [0,1]\\
7 & 0.001 & [0,0]\\
8 & 0.000 & [0,0]
\end{array}
\end{equation*}
\end{minipage}
\begin{minipage}[b][.18\textheight][c]{0.22\textwidth}
\begin{equation*}
\begin{array}{c|c|c}
r\, & \;\,\textrm{mean}\,\; & \;\;99\%\;\;\\
\hline
2 & 0.000 & [0,0]\\
3 & 1.201 & [1,3]\\
4 & 21.975 & [15,48]\\
5 & 0.345 & [0,5]\\
6 & 0.006 & [0,1]\\
7 & 0.000 & [0,0]\\
8 & 0.000 & [0,0]
\end{array}
\end{equation*}
\end{minipage}
\begin{minipage}[b][.18\textheight][c]{0.22\textwidth}
\begin{equation*}
\begin{array}{c|c|c}
r\, & \;\,\textrm{mean}\,\; & \;\;99\%\;\;\\
\hline
2 & 0.000 & [0,0]\\
3 & 1.164 & [1,3]\\
4 & 20.902 & [14,47]\\
5 & 0.067 & [0,1]\\
6 & 0.000 & [0,0]\\
7 & 0.000 & [0,0]\\
8 & 0.000 & [0,0]
\end{array}
\end{equation*}
\end{minipage}

\begin{minipage}[b][.18\textheight][c]{0.09\textwidth}
\begin{equation*}
d^\perp\!=\!4
\end{equation*}
\end{minipage}
\begin{minipage}[b][.18\textheight][c]{0.22\textwidth}
\begin{equation*}
\begin{array}{c|c|c}
r\, & \;\,\textrm{mean}\,\; & \;\;99\%\;\;\\
\hline
2 & 0.000 & [0,0]\\
3 & 0.000 & [0,0]\\
4 & 6.178 & [1,14]\\
5 & 6.514 & [5,20]\\
6 & 1.263 & [1,7]\\
7 & 0.035 & [0,1]\\
8 & 0.000 & [0,0]
\end{array}
\end{equation*}
\end{minipage}
\begin{minipage}[b][.18\textheight][c]{0.22\textwidth}
\begin{equation*}
\begin{array}{c|c|c}
r\, & \;\,\textrm{mean}\,\; & \;\;99\%\;\;\\
\hline
2 & 0.000 & [0,0]\\
3 & 0.000 & [0,0]\\
4 & 5.963 & [1,14]\\
5 & 1.882 & [1,8]\\
6 & 0.090 & [0,1]\\
7 & 0.000 & [0,0]\\
8 & 0.000 & [0,0]
\end{array}
\end{equation*}
\end{minipage}
\begin{minipage}[b][.18\textheight][c]{0.22\textwidth}
\begin{equation*}
\begin{array}{c|c|c}
r\, & \;\,\textrm{mean}\,\; & \;\;99\%\;\;\\
\hline
2 & 0.000 & [0,0]\\
3 & 0.000 & [0,0]\\
4 & 5.525 & [1,12]\\
5 & 0.357 & [0,5]\\
6 & 0.010 & [0,1]\\
7 & 0.001 & [0,0]\\
8 & 0.000 & [0,0]
\end{array}
\end{equation*}
\end{minipage}
\begin{minipage}[b][.18\textheight][c]{0.22\textwidth}
\begin{equation*}
\begin{array}{c|c|c}
r\, & \;\,\textrm{mean}\,\; & \;\;99\%\;\;\\
\hline
2 & 0.000 & [0,0]\\
3 & 0.000 & [0,0]\\
4 & 4.885 & [1,11]\\
5 & 0.053 & [0,1]\\
6 & 0.000 & [0,0]\\
7 & 0.000 & [0,0]\\
8 & 0.000 & [0,0]
\end{array}
\end{equation*}
\end{minipage}

\begin{minipage}[b][.18\textheight][c]{0.09\textwidth}
\begin{equation*}
d^\perp\!=\!5
\end{equation*}
\end{minipage}
\begin{minipage}[b][.18\textheight][c]{0.22\textwidth}
\begin{equation*}
\begin{array}{c|c|c}
r\, & \;\,\textrm{mean}\,\; & \;\;99\%\;\;\\
\hline
2 & 0.000 & [0,0]\\
3 & 0.000 & [0,0]\\
4 & 0.000 & [0,0]\\
5 & 5.847 & [5,15]\\
6 & 1.153 & [1,6]\\
7 & 0.020 & [0,1]\\
8 & 0.000 & [0,0]
\end{array}
\end{equation*}
\end{minipage}
\begin{minipage}[b][.18\textheight][c]{0.22\textwidth}
\begin{equation*}
\begin{array}{c|c|c}
r\, & \;\,\textrm{mean}\,\; & \;\;99\%\;\;\\
\hline
2 & 0.000 & [0,0]\\
3 & 0.000 & [0,0]\\
4 & 0.000 & [0,0]\\
5 & 1.485 & [1,6]\\
6 & 0.055 & [0,1]\\
7 & 0.001 & [0,0]\\
8 & 0.000 & [0,0]
\end{array}
\end{equation*}
\end{minipage}
\begin{minipage}[b][.18\textheight][c]{0.22\textwidth}
\begin{equation*}
\begin{array}{c|c|c}
r\, & \;\,\textrm{mean}\,\; & \;\;99\%\;\;\\
\hline
2 & 0.000 & [0,0]\\
3 & 0.000 & [0,0]\\
4 & 0.000 & [0,0]\\
5 & 0.197 & [0,2]\\
6 & 0.002 & [0,0]\\
7 & 0.000 & [0,0]\\
8 & 0.000 & [0,0]
\end{array}
\end{equation*}
\end{minipage}
\begin{minipage}[b][.18\textheight][c]{0.22\textwidth}
\begin{equation*}
\begin{array}{c|c|c}
r\, & \;\,\textrm{mean}\,\; & \;\;99\%\;\;\\
\hline
2 & 0.000 & [0,0]\\
3 & 0.000 & [0,0]\\
4 & 0.000 & [0,0]\\
5 & 0.033 & [0,1]\\
6 & 0.000 & [0,0]\\
7 & 0.000 & [0,0]\\
8 & 0.000 & [0,0]
\end{array}
\end{equation*}
\end{minipage}
\vspace{-\baselineskip}\caption{some experimental data on $\defect(\phi_r)$ for random $[56,16]_2$-codes}\label{statdef}
\vspace{-\baselineskip}\end{figure}

\vspace{\baselineskip}

Consequently, we postulate the following relaxation of the minimal resolution conjecture:

\begin{heuristic}\label{heurdef}
Fix a field cardinality $q$,
assume $n$ is not too close to $k$ in order to stay away from the counterexamples to the minimal resolution conjecture (e.g. $n>k+\sqrt{2k}+2$ should suffice),
and $n<\binom{k+1}{2}$ in order to ensure regularity $2$.
Then for random $[n,k]_q$-codes, with high probability:
\begin{enumerate}
\item if $d>k+1-\rcrit$ we expect $\beta_{r-1,r}=0$ for $r>\rcrit$
\item if $d^\perp>\rcrit$ we expect $\beta_{r-1,r}=\left(\frac{k(k+1)}{r}-n\right)\binom{k-1}{r-2}$ for $r<\rcrit$
\end{enumerate}
where $\rcrit=\frac{k(k+1)}{n}$.
\end{heuristic}
\begin{remark}\label{heurdefasympt}
Consider this Heuristic in the asymptotic regime.
Setting $R=k/n$,
we can take $d=\dist_{GV}(q,n,k)\approx H_q^{-1}(1-R)n$ and $d^\perp=\dist_{GV}(q,n,n-k)\approx H_q^{-1}(R)n$ the corresponding Gilbert-Varshamov distances,
where $H_q$ is the $q$-ary entropy function.
Then the condition in~1. translates as $H_q^{-1}(1-R)>R(1-R)$, and the condition in~2. translates as $H_q^{-1}(R)>R^2$,
both of which are satisfied when $R$ is small enough.
In particular for $q=2$, we find that~1. is satisfied for $R<0.277$ and~2. is satisfied for $R<0.141$.
\end{remark}

\section{The distinguisher}\label{paramdist}

\subsection*{Basic version (without shortening)}

Distinguishers typically work by computing certain code invariants.
We might have theoretical bounds on the values of these invariants, that are essential for an asymptotic analysis.
However these bounds need not be tight. Hence for a given set of finite parameters, we can also adopt a more empirical
approach: sample a certain number of codes, compute their invariants, and observe when the distinguisher ``just works''.

So for given $q,m,t$ and a type of codes (alternant or Goppa),
we set $n=q^m$ and we compute the Betti numbers of a certain number of dual codes.
Most of the time, it turns out that these numbers are the same for all samples. We will denote
by $\beta_{r-1,r}^*$ these common values.
By Corollary~\ref{puncture}, these $\beta_{r-1,r}^*$ still provide lower bounds on $\beta_{r-1,r}$ for smaller $n$.

For a given $r$ with $\beta_{r-1,r}^*>0$, our distinguisher then works as follows.
Taking as input a (dual) code $\code$ of dimension $k=mt$, compute $\beta_{r-1,r}(\code)$ with Algorithm~\ref{algobases}.
Then if $\beta_{r-1,r}(\code)\geq\beta_{r-1,r}^*$ declare $\code$ is of the special type in question.
Otherwise declare $\code$ is random.

Observe that for random codes of dimension $k$,
if the conditions in Heuristic~\ref{heurdef} are satisfied,
we expect $\beta_{r-1,r}=\left(\frac{k(k+1)}{r}-n\right)^+\binom{k-1}{r-2}$.
Thus the distinguisher succeeds when this value is smaller than $\beta_{r-1,r}^*$,
or equivalently when
\beq\label{seuil_dist}
n\geq \left\lceil\frac{k(k+1)}{r}-\frac{\beta_{r-1,r}^*-1}{\binom{k-1}{r-2}}\right\rceil.
\eeq
When $r=2$, this gives the distinguishability threshold of the square code distinguisher of~\cite{FGOPT13}.
Using syzygies of higher degree $r$ allows to make the term $\frac{k(k+1)}{r}$ smaller and reach
a broader range of parameters.


In the following paragraphs we compare our new distinguisher (in this basic version) with the one from~\cite{AC:CouMorTil23}, on the very same examples
taken from this reference: $\Gop^{\text{irr}}_{4,4,4}$ and $\Gop^{\text{irr}}_{2,6,3}$.
In both cases, we improve the distinguishability threshold.
Interestingly, for $\Gop^{\text{irr}}_{4,4,4}$, this improved threshold is the one given by \eqref{seuil_dist} (with $r=\rmax=4$),
i.e. the conditions in Heuristic~\ref{heurdef} are always met.
On the other hand, for $\Gop^{\text{irr}}_{2,6,3}$, the conditions in Heuristic~\ref{heurdef} are the limiting factor
(but still improving over \cite{AC:CouMorTil23}).

A general description of the set of distinguishable parameters is complicated, as it has to encompass
both \eqref{seuil_dist} and the conditions in the heuristic.
However, a very interesting phenomenon will appear in the asymptotic analysis: all these constraints will be
automatically satisfied, i.e. asymptotically all parameters will become distinguishable.

\subsection*{Distinguishing Goppa codes with $q=4$, $m=4$, $t=4$, irreducible Goppa polynomial}

For these parameters, the square code distinguisher works down to $n_{\text{square}}=\mathbf{97}$,
while the distinguisher of~\cite{AC:CouMorTil23} works down to $n_{\text{CMT}}=\mathbf{80}$.

We experimentally find that dual Goppa codes $\code$ with these parameters, with $n=q^m=256$,
consistently have $r_{\max}=4$, matching the nonbinary case at the end of Remark~\ref{improvedcorlowerboundGop}
($a=\lceil\log_q(t/(q-1))\rceil=1$, $r^*=(a+1)t-q^a=4$).
Moreover, their Betti numbers are:
\beq
\beta_{1,2}^*=40,\qquad\beta_{2,3}^*=80,\qquad\beta_{3,4}^*=12.
\eeq
We then observe that the $\beta_{2,3}$-distinguisher works down to $n_{\beta_{2,3}}=\mathbf{86}$,
and the $\beta_{3,4}$-distinguisher works down to $n_{\beta_{3,4}}=\mathbf{68}$,
both of which coincide with \eqref{seuil_dist}.
More precisely, computing the Betti numbers of dual Goppa and random codes around these values of $n$ consistently yields:
\begin{equation*}
\arraycolsep=5pt\def\arraystretch{1.2}
\begin{array}{|c|ccccccccccccc|}
\hline
n & \cdots & 88 & 87 & \mathbf{86} & 85 & 84 & \cdots & 70 & 69 & \mathbf{68} & 67 & 66 & \cdots\\
\hline
\beta_{2,3}^{\Gop^{\text{irr},\perp}} & & 80 & 80 & \mathbf{80} & 85 & 100 & & 310 & 325 & 340 & 355 & 370 & \\
\beta_{2,3}^{\text{random}} & & 40 & 55 & \mathbf{70} & 85 & 100 & & 310 & 325 & 340 & 355 & 370 & \\
\hline
\beta_{3,4}^{\Gop^{\text{irr},\perp}} & & 12 & 12 & 12 & 12 & 12 & & 12 & 12 & \mathbf{12} & 105 & 210 & \\
\beta_{3,4}^{\text{random}} & & 0 & 0 & 0 & 0 & 0 & & 0 & 0 & \mathbf{0} & 105 & 210 & \\
\hline
\end{array}
\end{equation*}
We see that they stick to their ``predicted'' values: $\beta_{r-1,r}=\max(\beta_{r-1,r}^*,\ind(\phi_r)^+)$
for dual Goppa, and $\beta_{r-1,r}=\ind(\phi_r)^+$ for random codes.

\subsection*{Distinguishing Goppa codes with $q=2$, $m=6$, $t=3$, irreducible Goppa polynomial}

For these parameters, the square code distinguisher works down to $n_{\text{square}}=\mathbf{62}$,
while the distinguisher of~\cite{AC:CouMorTil23} works down to $n_{\text{CMT}}=\mathbf{59}$.

We experimentally find that dual Goppa codes $\code$ with these parameters, with $n=q^m=64$,
consistently have $r_{\max}=8$, matching the value given in Remark~\ref{improvedcorlowerboundGop}
($a=\lceil\log_2(t))\rceil=2$, $r^*=(a+2)t-2^a=8$).
Moreover, their top Betti numbers are:
\beq
\beta_{5,6}^*=1020,\qquad\beta_{6,7}^*=288,\qquad\beta_{7,8}^*=42.
\eeq
From~\eqref{seuil_dist} we expect to distinguish
at $\beta_{5,6}$ for $n\geq57$,
at $\beta_{6,7}$ for $n\geq49$,
and at $\beta_{7,8}$ for $n\geq43$.

And indeed at $n_{\beta_{5,6}}=\mathbf{57}$ we consistently find $\beta_{5,6}\geq 1020$ for dual Goppa codes, while $\beta_{5,6}<500$ for random codes with quite high probability.

For smaller $n$ we have to pass to $\beta_{6,7}$.
The $\beta_{6,7}$-distinguisher works well for $n=56$, but the quality gradually falls
(distinguishing errors occur more frequently) as $n$ becomes smaller.
It is difficult to point a precise threshold where the distinguisher ceases to work.
Arguably we still have a positive advantage at $n_{\beta_{6,7}}=\mathbf{50}$, but not anymore at $n=49$.
We could then try with $\beta_{7,8}$, but this fails too.

What happens? It turns out the conditions in Heuristic~\ref{heurdef} are not satisfied anymore,
so we cannot ensure $\defect(\phi_r)=0$, or even $\defect(\phi_r)$ small, for these values of $r$ and $n$.
Indeed, in order to have $\beta_{r-1,r}=0$, we need $r>k+1-d$,
where $d=\dmin(\code)$. For $r=7$ and $k=mt=18$ this gives $d\geq 13$.
Then for smaller $d$, from Experimental fact~\ref{Betti_dmin_dduale},
we expect a loose link between $\beta_{r-1,r}$ and the weight distribution of $\code$.
Recall the distinguisher still works as long as random codes satisfy $\beta_{r-1,r}<\beta_{r-1,r}^*$ with high enough probablity.
Experimentally, for $r=7$, we find this inequality is satisfied for random codes of minimum distance $d\geq10$,
while $d=9$ is a borderline case, and $d\leq8$ fails invariably.
Now what happens is that, as $n$ decreases from $56$ to $50$, the proportion of random codes with $d\geq10$ also decreases,
and they become minority for $n=49$.

\begin{remark}
The sheer fact that the minimum distance is easily computable for codes with such small parameters actually
provides a much more efficient alternative distinguisher. Indeed, turning back to the primal codes, the designed distance
for binary Goppa codes with $t=3$ is $7$. Let us thus consider a distinguisher that decides that a given $[n,n-18]$-code
is Goppa if it has $\dmin\geq 7$, and that it is random otherwise. Experiments then show that, as long as $n\geq32$,
less than $2\%$ of random codes with these parameters have $\dmin\geq 7$. This means that this distinguisher has more
than $99\%$ success rate! And it remains remarkably effective even for smaller $n$.
For instance, as low as $n=24$, less than half random codes have $\dmin\geq 7$, so this distinguisher
still has more than $75\%$ success rate.

In these experiments, random codes were defined by taking a uniformly random $(n-18)\times n$ generator matrix,
which might quite often produce degenerate codes as $n$ approaches $18$.
Other probability distributions preventing this problem could be used, which would slightly alter the success rate,
but not the qualitative behaviour of this distance-based distinguisher.
\end{remark}

\subsection*{Full version (with shortening)}

We can also use \emph{shortened} dual codes as input in our distinguisher.
In general, shortening $s$ times a $[n,k]$-code $\code$ produces a $[n_s,k_s]$-code $\code_s$ with $n_s=n-s$, $k_s=k-s$.
Having smaller parameters will decrease the complexity. Moreover the (dual) rate $R_s=k_s/n_s$
also decreases when we shorten; as observed in Remark~\ref{heurdefasympt}, this will help satisfying
the conditions in Heuristic~\ref{heurdef}.

\begin{example}
Consider the class $\Gop^{\text{irr},\perp}_{2,10,10}$ of dual binary Goppa codes with $m=10$, $n=1024$, $t=10$, irreducible Goppa polynomial.
Without shortening, these codes are out of reach of our distinguisher:
we would have to run Algorithm~\ref{algobases} up to $r=10$
and do linear algebra in dimension roughly $10^{20}$, which is infeasible.
On the other hand, we can effectively distinguish after shortening $40$ times:
using the code in \cite{githubsyzygies} we experimentally find that $40$-shortened $\Gop^{\text{irr},\perp}_{2,10,10}$ have $\beta_{3,4}=30$ (Figure~\ref{Gop21010}),
while $\beta_{3,4}=0$ for random codes with the same parameters.
Computational cost (a couple of weeks each, on a workstation with several hundreds of GB of RAM)
prevented the author to run more than a few tests, but no deviation was observed.
In particular $\beta_{3,4}=0$ consistently in the random case confirms Heuristic~\ref{heurdef} for these parameters.
\end{example}

\begin{proposition}\label{propdistrac}
Fix a class $\mathcal{T}_{q,m,n,t}$ of dual alternant or Goppa codes with some given parameters $q,m,n,t$,
and assume that for some integer $r^*$ we have a lower bound of the form
\beq\label{hypothesermax}
\rmax(\code_s)\geq r^*-s
\eeq
for all $\code\in\mathcal{T}_{q,m,n,t}$ and all $s$-shortened subcodes $\code_s$ of $\code$.

Set $k=mt$, and let $\code$ be either a random $[n,k]_q$-code, or an element of $\mathcal{T}_{q,m,n,t}$.

Choose an integer $s$ such that
\beq\label{raccourci_admissible}
\rcrit(\code_s)=\frac{k_s(k_s+1)}{n_s}<r^*-s
\eeq
where $n_s=n-s$ and $k_s=k-s$.
Then set
\beq
r=r^*-s.
\eeq
\begin{enumerate}
\item Let $d_s\approx\dist_{GV}(q,n_s,k_s)$ be the typical minimum distance of random $[n_s,k_s]_q$-codes,
and assume $d_s>k_s+1-\rcrit(\code_s)$. Then, under part 1 of Heuristic~\ref{heurdef}, computing
\beq
\beta_{r-1,r}(\code_s)
\eeq
where $\code_s$ is any $s$-shortened subcode of $\code$,
allows to distinguish whether $\code$ is a random code (in which case $\beta_{r-1,r}(\code_s)=0$ w.h.p.)
or an element of $\mathcal{T}_{q,m,n,t}$ (in which case $\beta_{r-1,r}(\code_s)>0$).
\item Moreover, let $d_s^\perp\approx\dist_{GV}(q,n_s,n_s-k_s)$ be the typical dual distance of random $[n_s,k_s]_q$-codes,
and assume $d_s^\perp>\rcrit(\code_s)$. Then, under Heuristic~\ref{heurdef},
for large enough parameters,
the complexity of this distinguisher when Betti numbers are computed with the iterative Algorithm~\ref{algobases}
is dominated by the sum
\beq\label{kappasomme}
\kappa_{\operatorname{iter}}=\sum_{i\leq\left\lfloor\frac{k_s(k_s+1)}{n_s}\right\rfloor+1}\max\left(k_s\ind(\phi_{i-1}),\binom{k_s+1}{2}\ind(\phi_{i-2})\right)^\omega
\eeq
where $\ind(\phi_i)=\left(\frac{k_s(k_s+1)}{i}-n_s\right)\binom{k_s-1}{i-2}$
and $\omega\approx2.372$ is the exponent of linear algebra.
\end{enumerate}
\end{proposition}
\begin{proof}
If $\code$ is random, so is $\code_s$.
Then, under part 1 of Heuristic~\ref{heurdef}, condition~\eqref{raccourci_admissible}
and $d_s>k_s+1-\rcrit(\code_s)$ indeed imply $\beta_{r-1,r}(\code_s)=0$ for $r=r^*-s$.
On the other hand, for $\code$ in $\mathcal{T}_{q,m,n,t}$, hypothesis~\eqref{hypothesermax}
precisely means $\beta_{r-1,r}(\code_s)>0$.

Now using Algorithm~\ref{algobases} we obtain $\beta_{r-1,r}(\code_s)$ after computing the left kernel
of the matrices $\mat{M}_i$ for $2\leq i\leq r$.
The contributions for $i<4$ are negligible.
For $i\geq4$, the matrix $\mat{M}_i$ has size $k_s\beta_{i-2,i-1}(\code_s)\times\binom{k_s+1}{2}\beta_{i-3,i-2}(\code_s)$
and the contributions for $i>\lfloor\rcrit(\code_s)\rfloor+1$ are negligible because of part 1.
Then, as $d_s^\perp>\rcrit(\code_s)$, under part 2 of Heuristic~\ref{heurdef} we expect
$\beta_{i-1,i}(\code_s)=\ind(\phi_i)$ for $2\leq i\leq\rcrit(\code_s)$
for random $\code$; and
we can assume this holds also for $\code\in\mathcal{T}_{q,m,n,t}$, otherwise computing this dimension
readily provides a distinguisher with lower complexity. From this we conclude.
\end{proof}

Algorithm~\ref{algobases}, that iteratively computes the linear strand of the minimal resolution,
can be optimized using methods from sparse linear algebra and structured linear algebra,
as is done in~\cite{AlbLS01}. This might procure practical improvements, both on memory usage
and on time complexity, but the theoretical analysis of these improvements remains to be done.
A tentative implementation of the method of~\cite{AlbLS01} can be found in~\cite{githublinstrand}.
It would be very nice if this project could be continued, and maybe eventually incorporated in
some computer algebra systems.

\subsection*{An alternative approach through Koszul cohomology and the (block) Wiedemann algorithm}

Koszul cohomology \cite{Green84} allows to directly compute one syzygy space, with no need to iteratively
compute all the intermediate steps in the resolution.

Let $\code$ be a $[n,k]$-code, $S=\F[X_1,\dots,X_k]$
the polynomial ring in $k$ indeterminates,
and $\ev:S\longto\code\deux[\cdot]$ the graded evaluation map associated
to some basis $\mot{c}_1,\dots,\mot{c}_k$ of $\code$.

\begin{proposition}
For any $r$, the syzygy space $M_{r-1,r}(\code)$ can be computed as the cohomology of the sequence
\beq\label{courte}
\bigwedge\nolimits^rS_1\overset{\partial}{\longto}\left(\bigwedge\nolimits^{r-1}S_1\right)\tens\code\overset{\partial'}{\longto}\left(\bigwedge\nolimits^{r-2}S_1\right)\tens\code\deux
\eeq
i.e. as
\beq\label{kerd/imd'}
M_{r-1,r}(\code)=\ker(\partial')/\im(\partial)
\eeq
where the linear maps $\partial,\partial'$ are defined on the standard basis elements as
\beq\label{defd}
\begin{array}{cccc}
\partial: & \bigwedge^rS_1 & \longto & (\bigwedge^{r-1}S_1)\tens\code\\
 & X_{i_1}\wedge\cdots\wedge X_{i_r} & \mapsto & \sum_{j=1}^r(-1)^j(X_{i_1}\wedge\cdots\widehat{X}_{i_j}\cdots\wedge X_{i_r})\tens\mot{c}_{i_j}
\end{array}
\eeq
and
\beq\label{defd'}
\begin{array}{cccc}
\partial': & (\bigwedge^{r-1}S_1)\tens\code & \longto & (\bigwedge^{r-2}S_1)\tens\code\deux\\
 & (X_{i_1}\wedge\cdots\wedge X_{i_{r-1}})\tens\mot{c}_l & \mapsto & \sum_{j=1}^{r-1}(-1)^j(X_{i_1}\wedge\cdots\widehat{X}_{i_j}\cdots\wedge X_{i_{r-1}})\tens(\mot{c}_{i_j}\mot{c}_l)
\end{array}
\eeq
where notation $\widehat{X}_{i_j}$ means that $X_{i_j}$ is omitted.
\end{proposition}
\begin{proof}
Apply \cite[Th.(1.b.4)]{Green84} with $B=\code\deux[\cdot]$, $p=r-1$, $q=1$
(or \cite[Prop.~2.7]{Eisenbud05} with $M=\code\deux[\cdot]$, $i=r-1$, $j=r$).
\end{proof}

By construction $\partial,\partial'$ satisfy $\im(\partial)\subset\ker(\partial')$,
and elementary linear algebra allow to rephrase \eqref{kerd/imd'} as
\beq
M_{r-1,r}(\code)\simeq\ker(\partial'|_T)
\eeq
where $T$ is any complementary subspace to $\im(\partial)$ in $(\bigwedge^{r-1}S_1)\tens\code$.

Let us denote the standard basis of $(\bigwedge^{r-1}S_1)\tens\code$ by
\beq\label{basestandard}
\mot{e}_{i_1,\dots,i_{r-1};l}=(X_{i_1}\wedge\cdots\wedge X_{i_{r-1}})\tens\mot{c}_l
\eeq
for $1\leq i_1<\cdots<i_{r-1}\leq k$ and $1\leq l\leq k$.

\begin{lemma}\label{baseT}
The subfamily of the basis \eqref{basestandard} made of
those $\mot{e}_{i_1,\dots,i_{r-1};l}$
with $l\leq i_{r-1}$
is a basis of a complementary subspace $T$ to $\im(\partial)$ in $(\bigwedge^{r-1}S_1)\tens\code$.
\end{lemma}
\begin{proof}
First, a subfamily of a basis is linearly independent, hence a basis of its linear span $T$.

Next we show $\im(\partial)\cap T=0$. An element of $\im(\partial)$ is of the form
\beq
\sum_{1\leq i_1<\cdots<i_r\leq k}\lambda_{i_1,\dots,i_r}\sum_{1\leq j\leq r}(-1)^j\mot{e}_{i_1,\dots\widehat{i_j}\dots,i_r;i_j}.
\eeq
If such an element is in $T$, then the coefficient of each $\mot{e}_{i_1,\dots,i_{r-1};i_r}$ is $0$
(because $i_r>i_{r-1}$). However the coefficient of $\mot{e}_{i_1,\dots,i_{r-1};i_r}$ is $(-1)^r\lambda_{i_1,\dots,i_r}$.
Hence all $\lambda_{i_1,\dots,i_r}=0$.

To conclude, we show that $\im(\partial)+T$ contains all the $\mot{e}_{i_1,\dots,i_{r-1};l}$.
If $l\leq i_{r-1}$ we already have $\mot{e}_{i_1,\dots,i_{r-1};l}\in T$.
On the other hand, if $l>i_{r-1}$, then we can set $i_r=l$ and we find
$\mot{e}_{i_1,\dots,i_{r-1};l}=\partial(\mot{s})-\mot{t}$ where $\mot{s}=(-1)^rX_{i_1}\wedge\cdots\wedge X_{i_r}$
and where $\mot{t}=\sum_{j=1}^{r-1}(-1)^{r-j}\mot{e}_{i_1,\dots\widehat{i_j}\dots,i_r;i_j}\in T$
(because $i_j\leq i_r$ for $j\leq r-1$).
\end{proof}

\begin{lemma}We have
\beq\label{dimT}
\dim(T)=(r-1)\binom{k+1}{r}=\frac{k(k+1)}{r}\binom{k-1}{r-2}.
\eeq
\end{lemma}
\begin{proof}
We give two proofs.
For the first proof, we count the number of indices $i_1,\dots,i_{r-1};l$ in Lemma~\ref{baseT}.
Conditioning on $i_{r-1}=u$,
the number of $i_1,\dots,i_{r-2}$ with $1\leq i_1<\cdots<i_{r-2}<u$ is $\binom{u-1}{r-2}$
and the number of $l$ with $1\leq l\leq u$ is $u$. It follows $\dim(T)=\sum_{u\leq k}u\binom{u-1}{r-2}$,
and an easy manipulation with binomal coefficients gives \eqref{dimT}.

For the second proof, we observe that $\partial$ is injective, so that
$\dim(T)=\dim((\bigwedge^{r-1}S_1)\tens\code)-\dim(\bigwedge^rS_1)=k\binom{k}{r-1}-\binom{k}{r}$,
and we conclude likewise.
\end{proof}


\begin{remark}
The sequence \eqref{courte} continues with
\beq\label{longue}
\left(\bigwedge\nolimits^{r-2}S_1\right)\tens\code\deux\overset{\partial''}{\longto}\left(\bigwedge\nolimits^{r-3}S_1\right)\tens\code\deux[3]\longto\cdots\longto C\deux[r]\longto0
\eeq
and we have $\im(\partial')\subset U=\ker(\partial'')$.
Also identifying $T$ with $\coker(\partial)$, the map $\partial'$ then induces a map
\beq
\overline{\partial'}:T\longto U
\eeq
with $\dim\ker(\overline{\partial'})=\beta_{r-1,r}(\code)$ and $\dim\coker(\overline{\partial'})=\beta_{r-2,r}(\code)$.
Moreover, if $\code$ has regularity $2$, then:
\begin{itemize}
\item in \eqref{longue} we can replace $\code\deux,\code\deux[3],\dots,\code\deux[r]$ with $\F^n$
\item this sequence \eqref{longue} is acyclic.
\end{itemize}
This allows to compute
\beq
\dim(U)=n\sum_{i=2}^r(-1)^i\binom{k}{r-i}=n\binom{k-1}{r-2}.
\eeq
We then find
$\ind(\overline{\partial'})=\dim(T)-\dim(U)=\left(\frac{k(k+1)}{r}-n\right)\binom{k-1}{r-2}$
which, as expected, is the same as $\ind(\phi_r)$ given in Theorem~\ref{indphi}.
\end{remark}

\begin{remark}
On the other hand, if we compute the index of $\partial'|_T$ without corestriction of the target space,
we find, for $\code$ of regularity $2$,
\beq
\begin{split}
\ind(\partial'|_T)&=\dim(T)-\dim\left(\left(\bigwedge\nolimits^{r-2}S_1\right)\tens\F^n\right)\\
&=\frac{k(k+1)}{r}\binom{k-1}{r-2}-n\binom{k}{r-2}\;<\;\ind(\phi_r).
\end{split}
\eeq
In most applications we will have $\ind(\phi_r)<0$, hence also $\ind(\partial'|_T)<0$.
\end{remark}

\vspace{.5\baselineskip}

In the simplified version of our distinguisher, we do not need to actually compute the syzygy space $M_{r-1,r}=\ker(\partial'|_T)$,
nor its dimension $\beta_{r-1,r}$: we only want to decide whether it is $0$ or not.

Wiedemann's algorithm~\cite{Wiedemann86}, in its basic version, probabilistically finds a preimage of a nonsingular linear map $A$,
with only black-box queries to $A$.
This means that the overall cost of the algorithm does not involve the cost of general linear algebra,
but only the cost $c(A)$ of evaluation of $A$.
Several reductions provide variants of the algorithm that find a nonzero element in the kernel of a singular $A$,
at essentially the same cost.

This allows us to propose a ``Koszul-Wiedemann'' (KW) alternative implementation of the distinguisher.
We keep the setting of part 1 of Proposition~\ref{propdistrac}, in particular the integer $s$ should
satisfy \eqref{raccourci_admissible}, and the typical minimum distance $d_s$ should satisfy $d_s>k_s+1-\rcrit(\code_s)$
where $\rcrit(\code_s)=\frac{k_s(k_s+1)}{n_s}$.
Then we apply Wiedemann's algorithm to find a nonzero element in the kernel
of the linear map $\partial'|_T$ associated with $\code_s$ and $r=r^*-s$:
\begin{enumerate}[(a)]
\item if Wiedemann's algorithm succeeds, we declare $\code$ is in $\mathcal{T}_{q,m,n,t}$
\item otherwise, we declare $\code$ is probably random.
\end{enumerate}
\begin{proposition}\label{propcomplexiteKW}
Under part 1 of Heuristic~\ref{heurdef}, this procedure is correct in the following sense:
\begin{itemize}
\item in case (a), the code $\code$ indeed is in $\mathcal{T}_{q,m,n,t}$ with probability very close to $1$
\item in case (b), the code $\code$ is random with nonnegligible probability (and we can then iterate the procedure in order to make this probability become very close to $1$).
\end{itemize}
Each application of this procedure has complexity
\beq\label{kappaKW}
\kappa_{KW}=O\left(n_s^2\binom{k_s}{r-2}\left((r-1)^2\binom{k_s+1}{r}+\binom{k_s}{r-2}\right)\right).
\eeq
\end{proposition}
\begin{proof}
Correctness follows from the discussion just above, and part 1 of Proposition~\ref{propdistrac}.

Now consider the map $\partial'|_T$ associated with $\code_s$ and $r=r^*-s$:
\beq
\partial'|_T:T\longto(\bigwedge\nolimits^{r-2}S_1)\tens\F^{n_s}
\eeq
where $S$ now is the polynomial ring in $k_s$ indeterminates.
By~\eqref{dimT} the source space of this map has dimension
\beq
M=\dim(T)=(r-1)\binom{k_s+1}{r}
\eeq
while on the other end the target space has dimension
\beq
N=n_s\binom{k_s}{r-2}.
\eeq
Moreover, by formula \eqref{defd'}, this map sends each basis vector of $T$
to a vector of weight at most $(r-1)n_s$ in the standard basis of the target space.
Thus the cost of evaluation of $\partial'|_T$ satisfies
\beq
c(\partial'|_T)\leq M(r-1)n_s.
\eeq
Now the estimate for the complexity of Wiedemann's algorithm is
\beq
O(N(c(\partial'|_T)+N))
\eeq
and we conclude.
\end{proof}

It is also possible to use the block version of Wiedemann's algorithm \cite{Coppersmith94}.
This definitely improves the complexity, but does not qualitatively change the theoretical
analysis~\eqref{kappaKW} as this improvement can be incorporated in the constant in the big O notation.

Two tentative implementations of the KW version of the distinguisher can be found in~\cite{githubkoszul},
one with the original Wiedemann algorithm, one with the block Wiedemann algorithm.
It was especially tempting to leverage the very optimized implementation of block Wiedemann
that can be extracted from~\cite{cadonfs}. However several technical difficulties appear,
such as the fact that this implementation of block Wiedemann was not designed to make black-box queries to the
linear map, or also the fact that it uses heuristics adapted to the specifics of the linear
systems appearing in the number field sieve, while these heuristics can lead to an unexpected
behaviour in our use case.

\subsection*{Application to the \Classic parameters}

The following two tables show how the distinguishers from Propositions~\ref{propdistrac} and~\ref{propcomplexiteKW}
apply to \Classic parameters.

In Figure~\ref{tablef} we use for $r^*$ the lower bound $\widehat{f}$ proven in Corollary~\ref{corlowerboundGop},
and then choose the maximum possible shortening order $s$ satisfying~\eqref{raccourci_admissible}.
We observe that the KW approach provides a noticeable complexity improvement over the iterative approach.

In Figure~\ref{tabler} we do the same thing but using the tight lower bound on $r^*$ stated in Remark~\ref{improvedcorlowerboundGop}.
This provides a further improvement, for both approaches.

The reader should keep in mind that the
expression used for $\kappa_{\operatorname{iter}}$ assumes that linear algebra in dimension $n$ can be done
in complexity $n^\omega$, which only holds asymptotically.
Likewise the value used for $\kappa_{KW}$ in the tables is the inside of the big O in formula~\eqref{kappaKW}
so that some constant factor is discarded.
More fundamentally, we avoided defining a precise model of computation.
As a consequence, the exact numbers given here should not be taken too literally.
But still, they give a reasonably faithful idea of what is happening.
In particular we see that these complexity estimates significantly improve those from~\cite{AC:CouMorTil23},
although they remain practically unreacheable and well beyond security levels.

\begin{figure}[h!]
\begin{equation*}
\renewcommand*{\arraystretch}{1.1}
\begin{array}{|c|c|c|c|c|c|}
\hline
(n,m,t) & (3488,12,64) & (4608,13,96) & (6688,13,128) & (6960,13,119) & (8192,13,128)\\
\hline
\widehat{f}\rule{0pt}{10pt} & 427 & 683 & 939 & 867 & 939\\
\hline
s & 377 & 568 & 816 & 769 & 848\\
\hline
r & 50 & 115 & 123 & 98 & 91\\
\hline
[n_s,k_s] & [3111,391] & [4040,680] & [5872,848] & [6191,778] & [7344,816]\\
\hline
\rcrit & 49.27 & 114.62 & 122.61 & 97.89 & 90.78\\
\hline
\dist_{GV} & 921 & 1069 & 1650 & 1828 & 2256\\
\hline
\dist_{GV}^\perp & 55 & 102 & 122 & 108 & 110\\
\hline
\kappa_{\operatorname{iter}}\rule{0pt}{10pt} & 2^{528} & (2^{1080}) & (2^{1224}) & 2^{1030} & 2^{997}\\
\hline
\kappa_{KW}\rule{0pt}{10pt} & 2^{452} & 2^{916} & 2^{1038} & 2^{874} & 2^{847}\\
\hline
\end{array}
\end{equation*}
\vspace{-1.2\baselineskip}
\caption{complexity estimates using the lower bound $\widehat{f}$ for $r^*$}\label{tablef}
\end{figure}

\begin{figure}[h!]
\begin{equation*}
\renewcommand*{\arraystretch}{1.1}
\begin{array}{|c|c|c|c|c|c|}
\hline
(n,m,t) & (3488,12,64) & (4608,13,96) & (6688,13,128) & (6960,13,119) & (8192,13,128)\\
\hline
r^* & 448 & 736 & 1024 & 943 & 1024\\
\hline
s & 405 & 643 & 930 & 866 & 954\\
\hline
r & 43 & 93 & 94 & 77 & 70\\
\hline
[n_s,k_s] & [3083,363] & [3965,605] & [5758,734] & [6094,681] & [7238,710]\\
\hline
\rcrit & 42.86 & 92.47 & 93.69 & 76.21 & 69.74\\
\hline
\dist_{GV} & 933 & 1092 & 1693 & 1870 & 2306\\
\hline
\dist_{GV}^\perp & 50 & 89 & 103 & 92 & 93\\
\hline
\kappa_{\operatorname{iter}}\rule{0pt}{10pt} & 2^{469} & (2^{909}) & 2^{982} & 2^{841} & 2^{801}\\
\hline
\kappa_{KW}\rule{0pt}{10pt} & \mathbf{2^{401}} & \mathbf{2^{772}} & \mathbf{2^{834}} & \mathbf{2^{716}} & \mathbf{2^{683}}\\
\hline
\end{array}
\end{equation*}
\vspace{-1.2\baselineskip}
\caption{complexity estimates using the tight $r^*$}\label{tabler}
\end{figure}

In all instances we observe that the condition $\dist_{GV}(n_s,k_s)>k_s+1-\rcrit(\code_s)$ is satisfied
with a large margin.
This means that if the reader is reluctant to use part 1 of Heuristic~\ref{heurdef} with $r$ so close to $\rcrit$,
then it would actually be possible to choose a smaller $s$ and still distinguish (at the cost of a slightly worse complexity).

Concerning part 2 of Heuristic~\ref{heurdef}, we see that the condition $\dist_{GV}(n_s,n_s-k_s)>\rcrit(\code_s)$
is satisfied in both tables for the parameter sets $(3488,12,64)$, $(6960,13,119)$ and $(8192,13,128)$.
For $(6688,13,128)$ the condition holds when we use the tight estimate for $r^*$.
For $(4608,13,96)$ it fails, so there might be some doubt on the validity of $\kappa_{\operatorname{iter}}$ in this particular case.
On the other end we recall that Proposition~\ref{propcomplexiteKW} does not rely on this part of the heuristic,
hence $\kappa_{KW}$ is unaffected.

\subsection*{Asymptotics}

Fix a base field cardinality $q$, for instance $q=2$, and a (dual) rate $R$.
In \cite{rationale22} it is suggested to take a primal code of rate between $0.7$ and $0.8$,
so passing to the dual gives $0.2\leq R\leq 0.3$. However here we allow \emph{any} $R$.
Then for $n\to\infty$ set:
\begin{itemize}
\item $m=\lceil\log_q(n)\rceil=\log_q(n)+O(1)$
\item $k\approx Rn\;$ such that:
\item $t=\frac{k}{m}\;$ is an integer.
\end{itemize}
A key observation then is that the lower bound $f$ on $\rmax$ of dual alternant codes in Theorem~\ref{thlowerboundalt},
or equivalently, the number $f$ of columns of the matrix $\mat{\Phi}$ in \eqref{defPhi},
is very close to $k$:
\begin{lemma}\label{fgrand}
We have
\beq
f=\left(1-\frac{\log_q\log_q(n)}{\log_q(n)}+O\left(\frac{1}{\log_q(n)}\right)\right)k.
\eeq
\end{lemma}
\begin{proof}
Direct consequence of:
\begin{itemize}
\item $e\;=\;\lfloor\log_q(t-1)\rfloor\;=\;\log_q(n)-\log_q\log_q(n)+O(1)$
\item $f\;=\;(e+1)t-\frac{q^{e+1}-1}{q-1}\;=\;et+O(t)$.
\end{itemize}
\end{proof}
In the case of binary Goppa codes, the $\widehat{f}$ from Corollary~\ref{corlowerboundGop}
and the $r^*$ from Remark~\ref{improvedcorlowerboundGop} improve
the $f$ from Theorem~\ref{thlowerboundalt}.
However one could show that they all have the same asymptotics. So, in order to distinguish these codes
from random codes, we will not care about the extra structure: treating Goppa codes as alternant codes will suffice.
(Still, this could be used to distinguish Goppa from ``general'' alternant codes.)


Now, under Heuristic~\ref{heurdef}, we have:
\begin{theorem}\label{thcomplexiteasympt}
Asymptotically, $q$-ary alternant (including Goppa) codes of dual rate $R$ can be distinguished from random codes,
with complexity at most
\beq\label{complexite}
\kappa=q^{\left(\frac{2R^2}{1-R}+o(1)\right)\frac{(\log_q\log_q(n))^3}{(\log_q(n))^2}n}.
\eeq
\end{theorem}
\begin{proof}
We use part 1 of Proposition~\ref{propdistrac} with $\mathcal{T}=\Alt^\perp$, and $r^*=f$ as provided by Theorem~\ref{thlowerboundalt}.
The analysis will be easier if rephrased in terms of $k_{r^*}=k-r^*$
and $r=r^*-s$.
First, Lemma~\ref{fgrand} gives
\beq
k_{r^*}\;=\;k-f\;\sim\; R\frac{\log_q\log_q(n)}{\log_q(n)}n.
\eeq
Let us then turn to condition~\eqref{raccourci_admissible}. Fix an $\epsilon>0$,
and set
\beq
r=\left\lceil(1+\epsilon)\frac{R^2}{1-R}\left(\frac{\log_q\log_q(n)}{\log_q(n)}\right)^2n\right\rceil\ll k_{r^*}
\eeq
and $s=r^*-r$. Then $k_s=k_{r^*}+r$ and $n_s=n-k+k_{r^*}+r$, so for $n$ large enough we get
\beq\textstyle
\frac{k_s(k_s+1)}{n_s}=\frac{(k_{r^*}+r)(k_{r^*}+r+1)}{n-k+k_{r^*}+r}\approx\frac{k_{r^*}^2}{n-k}\approx\frac{R^2}{1-R}\left(\frac{\log_q\log_q(n)}{\log_q(n)}\right)^2n<r=r^*-s.
\eeq
Thus, for any $\epsilon>0$, condition~\eqref{raccourci_admissible} is satisfied with this value of $r$ as soon as $n$ is large enough, say larger than some $n(\epsilon)$.
But this means precisely that as $n\to\infty$, we can let $\epsilon=\epsilon(n)\to0$ while still satisfying the condition.
We then get
\beq
r\;\sim\;\frac{R^2}{1-R}\left(\frac{\log_q\log_q(n)}{\log_q(n)}\right)^2n.
\eeq
Moreover the shortened code $\code_s$ has rate $\frac{k_s}{n_s}=\frac{k_{r^*}+r}{n-k+k_{r^*}+r}=o(1)$,
so by Remark~\ref{heurdefasympt} both conditions in Heuristic~\ref{heurdef} are satisfied.
Thus the hypotheses of Proposition~\ref{propdistrac} are satisfied, and both the iterative and the KW version of the distinguisher
work with these parameters. It turns out the KW approach gives a slightly better asymptotic complexity.

Now in $\eqref{kappaKW}$ we see that, up to factors of lesser order, this complexity is essentially $\binom{k_s}{r}^2=\binom{k_{r^*}+r}{r}^2$.
Using $r\ll k_{r^*}$ and Stirling's formula we find $\log_q\binom{k_{r^*}+r}{r}\sim\log_q\binom{k_{r^*}}{r}\sim r\log_q\left(\frac{k_{r^*}}{r}\right)\sim r\log_q\log_q(n)$, and we conclude.
\end{proof}

In \cite[\S3.4]{security22}, the security level $2^b$
--- or less formally, the number $b$ of targeted ``bits of security'' ---
of the \Classic system, essentially defined as the complexity of the currently best attack (namely, by information set decoding),
is shown asymptotically to be linear in $t\propto n/\log(n)$.
Then in~\eqref{complexite} we have
\beq
\frac{\log(\kappa)}{n/\log(n)}\propto\frac{(\log\log(n))^3}{\log(n)}\to 0
\eeq
which means that the complexity of our distinguisher is subexponential in $b$
(although, admittedly, only very slightly so).

\vspace{.5\baselineskip}

As already mentioned, the existence of our subexponential distinguisher does not logically result
in a downgrade of the security of the McEliece cryptosystem and its variants.
All we can conclude is that the assumption that the original McEliece cryptosystem is OW-CPA
is an ad hoc assumption, certainly backed by decades of practical experience, but that cannot
easily be reduced to more fundamental theoretical justifications.

It also leaves open the question of whether some of the ideas introduced here,
possibly combined with other arguments, could lead to a subexponential key recovery algorithm.

\begin{credits}
\subsubsection{\ackname}
The author learned a lot from the code-based cryptography seminar organized by Jean-Pierre Tillich at Inria Paris.
More important perhaps, without his friendly pressure, probably this work would have stayed eternally in limbo.
He should be thanked for all that.
The very favourable environment towards research at ANSSI in general,
and at LCR in particular, was also highly helpful.
Henri~Gilbert's advices led to improvements in several aspects of the presentation.
Anonymous members of LCR unsuccessfully suggested the title ``The Dyztynguysher'' for this work.

The author is part of the ANR-21-CE39-0009~BARRACUDA and ANR-21-CE39-0008~KLEPTOMANIAC research projects.
Some of the computations were done on the \texttt{lame} servers in Télécom Paris.

The use of Koszul cohomology in this updated version of the paper originates from a suggestion by Wouter Castryck.
The author also would like to thank the Eurocrypt 2025 anonymous referees for their comments.
\end{credits}

\bibliography{abbrev2,crypto,MR,perso}
\bibliographystyle{splncs04}

\end{document}